\documentclass[sigconf,nonacm]{acmart}

\usepackage{algorithm}
\usepackage[noend]{algpseudocode}
\usepackage{acronym}
\usepackage{bookmark}
\usepackage{tikz}
\usepackage{caption}
\usepackage[capitalise,noabbrev]{cleveref}
\usepackage[inline]{enumitem}
\usepackage{subcaption}
\usepackage{pgfpages}
\usepackage{amsmath}
\usepackage{multirow}
\usepackage{makecell}
\usepackage{xcolor}
\usepackage{balance}
\usepackage{mathtools}
\usepackage{xifthen}

\usepackage{lipsum}
\usepackage{todonotes}
\setlipsum{language=english}

\acrodef{mal} [mIS] {Maximal Independent Set}
\acrodef{mis} [MIS] {Maximum Independent Set}
\acrodef{mni} [MNI] {Mininum Image}
\acrodef{nav} [nAV] {non Articulation Vertex}
\acrodef{flexis} [FLEXIS] {FLExible Frequent Subgraph Mining using MaXimal Independent Set}
\acrodef{fsm} [FSM] {Frequent Subgraph Mining}
\acrodef{support} [$\sigma$] {\textit{support}}
\acrodef{lerp} [$\lambda$] {\textit{slider}}

\definecolor{todobackground}{RGB}{255,255,179} 
\definecolor{todonebackground}{RGB}{152,251,152} 
\definecolor{checkbackground}{RGB}{180,225,255} 
\definecolor{changedTextColor}{RGB}{128, 0, 0} 
\definecolor{mehtaComment}{RGB}{128, 0, 0}
\definecolor{akshitComment}{RGB}{0, 0 128}
\tikzset{
  redNode/.style={draw, circle, fill=red!20, minimum size=8pt, inner sep=0},
  blueNode/.style={draw, circle, fill=blue!20, minimum size=8pt, inner sep=0},
  greenNode/.style={draw, circle, fill=green!40, minimum size=8pt, inner sep=0},
  yellowNode/.style={draw, circle, fill=yellow!20, minimum size=8pt, inner sep=0},
  markedNode/.style={draw, circle, fill=gray!20, dotted, minimum size=8pt, inner sep=0},
  edgeNode/.style={draw, circle, fill=black!10, minimum size=8pt, inner sep=0},
  singleEdge/.style={},
  doubleArrow/.style={<->, >=stealth},
  singleArrow/.style={->, >=stealth},
  flowArrow/.style={->, >=stealth, dashed, thick},
  clusterNode/.style={
    inner sep=0pt, circular sector, shape border uses incircle,
  },
  circularNode/.style={
    inner sep=0pt, circle, shape border uses incircle,
  },
}

\usetikzlibrary{backgrounds, graphs, shapes.symbols, fit, shapes.multipart, shapes.geometric, external, calc} 

\newcommand{\qoutes}[1]{``#1''}
\newcommand{\sqoutes}[1]{`#1'}
\newcommand{\edgepath}[0]{\leadsto}
\newcommand{\gPath}[2]{#1 \edgepath #2}
\newcommand{\set}[1]{\{#1\}}

\newcommand{\abs}[1]{\lvert#1\rvert}

\newcommand{\FunctionCall}[2]{\textproc{#1}(#2)}
\newcommand{\sForAll}[2]{\ForAll{#1}#2\EndFor} 
\newcommand{\sWhile}[2]{\While{#1}#2\EndWhile} 
\newcommand{\sIf}[2]{\If{#1}#2\EndIf}          
\newcommand{\OR}{\textbf{ or }}
\newcommand{\Not}{\textbf{ not }}

\newcommand{\runningDatagraph}[1]{
  \begin{tikzpicture}

    \def\smallerScale{16}
    \def\biggerScale{20}
    \def\smallerWidth{4}
    \def\biggerWidth{6}

    \node[redNode] (d1) at (-1.3, 0) {$d_1$};
    \node[redNode] (d2) at (-0.4, 0) {$d_2$};
    \node[redNode] (d3) at (0.4, 0) {$d_3$};
    \node[redNode] (d4) at (1.3, 0) {$d_4$};

    \node[blueNode] (d5) at (-0.8, 1.0) {$d_5$};
    \node[blueNode] (d6) at (0, 1.0) {$d_6$};
    \node[blueNode] (d7) at (0.8, 1.0) {$d_7$};

    \draw[doubleArrow] (d1) -- (d2);
    \draw[doubleArrow] (d2) -- (d3);
    \draw[doubleArrow] (d3) -- (d4);

    \draw[doubleArrow] (d5) -- (d6);
    \draw[doubleArrow] (d6) -- (d7);

    \draw[doubleArrow] (d1) -- (d5);
    \draw[doubleArrow] (d2) -- (d5);

    \draw[doubleArrow] (d2) -- (d6);
    \draw[doubleArrow] (d3) -- (d6);

    \draw[doubleArrow] (d3) -- (d7);
    \draw[doubleArrow] (d4) -- (d7);

    \begin{scope}[on background layer]
      \ifnum #1=1
        \node[circularNode, fill=blue!40, scale=\biggerScale] at (d1) {};
        \node[circularNode, fill=blue!40, scale=\biggerScale] at (d5) {};
        \node[circularNode, fill=blue!40, scale=\biggerScale] at (d2) {};
        \draw[singleEdge, fill=blue!40, draw=blue!40, line width=\biggerWidth] (d1) -- (d5);
        \draw[singleEdge, fill=blue!40, draw=blue!40, line width=\biggerWidth] (d2) -- (d5);
        \draw[singleEdge, fill=blue!40, draw=blue!40, line width=\biggerWidth] (d2) -- (d1);
        \node[circularNode, fill=blue!40, scale=\biggerScale] at (d3) {};
        \node[circularNode, fill=blue!40, scale=\biggerScale] at (d4) {};
        \node[circularNode, fill=blue!40, scale=\biggerScale] at (d7) {};
        \draw[singleEdge, fill=blue!40, draw=blue!40, line width=\biggerWidth] (d3) -- (d4);
        \draw[singleEdge, fill=blue!40, draw=blue!40, line width=\biggerWidth] (d7) -- (d4);
        \draw[singleEdge, fill=blue!40, draw=blue!40, line width=\biggerWidth] (d7) -- (d3);
      \fi

      \ifnum #1=2
        \node[circularNode, fill=blue!40, scale=\biggerScale] at (d3) {};
        \node[circularNode, fill=blue!40, scale=\biggerScale] at (d6) {};
        \node[circularNode, fill=blue!40, scale=\biggerScale] at (d2) {};
        \draw[singleEdge, fill=blue!40, draw=blue!40, line width=\biggerWidth] (d3) -- (d6);
        \draw[singleEdge, fill=blue!40, draw=blue!40, line width=\biggerWidth] (d2) -- (d6);
        \draw[singleEdge, fill=blue!40, draw=blue!40, line width=\biggerWidth] (d2) -- (d3);
      \fi

      \ifnum #1=3
        \node[circularNode, fill=blue!40, scale=\biggerScale] at (d1) {};
        \node[circularNode, fill=blue!40, scale=\biggerScale] at (d5) {};
        \node[circularNode, fill=blue!40, scale=\biggerScale] at (d2) {};
        \node[circularNode, fill=blue!40, scale=\biggerScale] at (d6) {};
        \draw[singleEdge, fill=blue!40, draw=blue!40, line width=\biggerWidth] (d1) -- (d5);
        \draw[singleEdge, fill=blue!40, draw=blue!40, line width=\biggerWidth] (d2) -- (d5);
        \draw[singleEdge, fill=blue!40, draw=blue!40, line width=\biggerWidth] (d2) -- (d1);
        \draw[singleEdge, fill=blue!40, draw=blue!40, line width=\biggerWidth] (d2) -- (d6);
        \draw[singleEdge, fill=blue!40, draw=blue!40, line width=\biggerWidth] (d5) -- (d6);
      \fi

      \ifnum #1=4
        \node[circularNode, fill=blue!40, scale=\biggerScale] at (d3) {};
        \node[circularNode, fill=blue!40, scale=\biggerScale] at (d7) {};
        \node[circularNode, fill=blue!40, scale=\biggerScale] at (d2) {};
        \node[circularNode, fill=blue!40, scale=\biggerScale] at (d6) {};
        \draw[singleEdge, fill=blue!40, draw=blue!40, line width=\biggerWidth] (d3) -- (d7);
        \draw[singleEdge, fill=blue!40, draw=blue!40, line width=\biggerWidth] (d6) -- (d3);
        \draw[singleEdge, fill=blue!40, draw=blue!40, line width=\biggerWidth] (d2) -- (d3);
        \draw[singleEdge, fill=blue!40, draw=blue!40, line width=\biggerWidth] (d2) -- (d6);
        \draw[singleEdge, fill=blue!40, draw=blue!40, line width=\biggerWidth] (d7) -- (d6);
      \fi

      \ifnum #1=5
      \node[circularNode, fill=blue!40, scale=\biggerScale] at (d1) {};
      \node[circularNode, fill=blue!40, scale=\biggerScale] at (d5) {};
      \node[circularNode, fill=blue!40, scale=\biggerScale] at (d2) {};
      \draw[singleEdge, fill=blue!40, draw=blue!40, line width=\biggerWidth] (d1) -- (d5);
      \draw[singleEdge, fill=blue!40, draw=blue!40, line width=\biggerWidth] (d2) -- (d5);
      \draw[singleEdge, fill=blue!40, draw=blue!40, line width=\biggerWidth] (d2) -- (d1);
      \fi

      \ifnum #1=6
      \node[circularNode, fill=blue!40, scale=\biggerScale] at (d3) {};
      \node[circularNode, fill=blue!40, scale=\biggerScale] at (d4) {};
      \node[circularNode, fill=blue!40, scale=\biggerScale] at (d7) {};
      \draw[singleEdge, fill=blue!40, draw=blue!40, line width=\biggerWidth] (d3) -- (d4);
      \draw[singleEdge, fill=blue!40, draw=blue!40, line width=\biggerWidth] (d7) -- (d4);
      \draw[singleEdge, fill=blue!40, draw=blue!40, line width=\biggerWidth] (d7) -- (d3);
      \fi

      \ifnum #1=7
      \node[circularNode, fill=blue!40, scale=\biggerScale] at (d1) {};
      \node[circularNode, fill=blue!40, scale=\biggerScale] at (d5) {};
      \node[circularNode, fill=blue!40, scale=\biggerScale] at (d2) {};
      \draw[singleEdge, fill=blue!40, draw=blue!40, line width=\biggerWidth] (d1) -- (d5);
      \draw[singleEdge, fill=blue!40, draw=blue!40, line width=\biggerWidth] (d2) -- (d5);
      \draw[singleEdge, fill=blue!40, draw=blue!40, line width=\biggerWidth] (d2) -- (d1);
      \node[circularNode, fill=blue!40, scale=\biggerScale] at (d3) {};
      \node[circularNode, fill=blue!40, scale=\biggerScale] at (d6) {};
      \node[circularNode, fill=blue!40, scale=\biggerScale] at (d2) {};
      \draw[singleEdge, fill=blue!40, draw=blue!40, line width=\biggerWidth] (d3) -- (d6);
      \draw[singleEdge, fill=blue!40, draw=blue!40, line width=\biggerWidth] (d2) -- (d6);
      \draw[singleEdge, fill=blue!40, draw=blue!40, line width=\biggerWidth] (d2) -- (d3);
      \node[circularNode, fill=blue!40, scale=\biggerScale] at (d3) {};
      \node[circularNode, fill=blue!40, scale=\biggerScale] at (d4) {};
      \node[circularNode, fill=blue!40, scale=\biggerScale] at (d7) {};
      \draw[singleEdge, fill=blue!40, draw=blue!40, line width=\biggerWidth] (d3) -- (d4);
      \draw[singleEdge, fill=blue!40, draw=blue!40, line width=\biggerWidth] (d7) -- (d4);
      \draw[singleEdge, fill=blue!40, draw=blue!40, line width=\biggerWidth] (d7) -- (d3);
      \fi
    \end{scope}

  \end{tikzpicture}
}

\newcommand{\runningDistinctPatterns}[5]{
    \def\nodeDistance{1} 
    \def\nodeSize{4} 
    \def\textOffset{0.25cm} 

    \pgfmathsetmacro{\offsetX}{#3} 
    \pgfmathsetmacro{\offsetY}{#4} 

    \tikzset{
      every node/.style={minimum size=\nodeSize, inner sep=0pt}
    }

    \coordinate (A) at (\offsetX, \offsetY);
    \coordinate (B) at (\offsetX + \nodeDistance, \offsetY);
    \coordinate (C) at ($(A)!0.5!(B)!1*sin(60)*1cm!90:(B)$); 
    \coordinate (C1) at (\offsetX + \nodeDistance, \offsetY + \nodeDistance);
    \coordinate (D1) at (\offsetX, \offsetY + \nodeDistance);

    \coordinate (MidAB) at ($(A)!0.5!(B)$);

    \node[below=\textOffset] at (MidAB) {#5};

    \ifcase#1\relax
    \or 
      \node[redNode]    (A) at (A) {$1$};
      \node[greenNode]  (B) at (B) {$2$};
      \node[blueNode]   (C) at (C) {$3$};
    \or 
      \node[redNode]    (A) at (A) {$1$};
      \node[greenNode]  (B) at (B) {$2$};
      \node[yellowNode] (C) at (C) {$3$};
    \or 
      \node[redNode]    (A) at (A) {$1$};
      \node[blueNode]   (B) at (B) {$2$};
      \node[yellowNode] (C) at (C) {$3$};
    \or 
      \node[greenNode]  (A) at (A) {$1$};
      \node[blueNode]   (B) at (B) {$2$};
      \node[yellowNode] (C) at (C) {$3$};
   \or 
      \node[redNode] (A) at (A) {$1$};
      \node[greenNode] (B) at (B) {$2$};
      \node[blueNode] (C1) at (C1) {$3$};
      \node[yellowNode] (D1) at (D1) {$3$};
      \draw[doubleArrow] (A) -- (B);
      \draw[doubleArrow] (B) -- (C1);
      \draw[doubleArrow] (D1) -- (A);
      \draw[doubleArrow] (A) -- (C1); 
      \draw[doubleArrow] (B) -- (D1); 
    \or 
      \node[redNode] (A) at (A) {$1$};
      \node[greenNode] (B) at (B) {$2$};
      \node[blueNode] (C1) at (C1) {$3$};
      \node[yellowNode] (D1) at (D1) {$3$};
      \draw[doubleArrow] (A) -- (B);
      \draw[doubleArrow] (B) -- (C1);
      \draw[doubleArrow] (C1) -- (D1);
      \draw[doubleArrow] (D1) -- (A);
      \draw[doubleArrow] (A) -- (C1); 
      \draw[doubleArrow] (B) -- (D1); 
    \else 
      \node[text=red] at ($(A)!0.5!(B)$) {Invalid index};
    \fi

    \ifnum#1<5
      \ifnum#2=0
        \draw[doubleArrow] (A) -- (B);
        \draw[doubleArrow] (B) -- (C);
        \draw[doubleArrow] (C) -- (A);
      \else
        \ifnum#2=1 
          \draw[doubleArrow] (B) -- (C);
        \else
          \ifnum#2=2 
            \draw[doubleArrow] (C) -- (A);
          \else
            \ifnum#2=3 
              \draw[doubleArrow] (A) -- (B);
            \fi
          \fi
        \fi
      \fi
    \fi
}

\makeatletter
\newcommand{\crefnames}[3]{%
  \@for\next:=#1\do{%
    \expandafter\crefname\expandafter{\next}{#2}{#3}%
  }%
}
\makeatother

\newcommand{\changed}[1]{\textcolor{black}{#1}}
\newcommand{\added}[1]{\textcolor{black}{#1}}

\title{FLEXIS: \textbf{FLEX}ible Frequent Subgraph Mining using Maximal \textbf{I}ndependent \textbf{S}ets}

\author{Akshit Sharma}
\affiliation{
  \institution{Colorado School of Mines}
  \city{Golden}
  \country{USA}
}
\email{akshitsharma@mines.edu}

\author{Sam Reinehr}
\affiliation{
  \institution{Colorado School of Mines}
  \city{Golden}
  \country{USA}
}
\email{sam@reinehr.dev}

\author{Dinesh Mehta}
\affiliation{
  \institution{Colorado School of Mines}
  \city{Golden}
  \country{USA}
}
\email{dmehta@mines.edu}

\author{Bo Wu}
\affiliation{
  \institution{Colorado School of Mines}
  \city{Golden}
  \country{USA}
}
\email{bwu@mines.edu}

\begin{document}

\begin{abstract}
  Frequent Subgraph Mining (FSM) is the process of identifying common subgraph patterns that surpass a predefined frequency threshold. While FSM is widely applicable in fields like bioinformatics, chemical analysis, and social network anomaly detection, its execution remains time-consuming and complex. This complexity stems from the need to recognize high-frequency subgraphs and ascertain if they exceed the set threshold. Current approaches to identifying these patterns often rely on edge or vertex extension methods. However, these strategies can introduce redundancies and cause increased latency. To address these challenges, this paper introduces a novel approach for identifying potential $k$-vertex patterns by combining two frequently observed $(k-1)$-vertex patterns. This method optimizes the breadth-first search, which allows for quicker search termination based on vertices count and support value. Another challenge in FSM is the validation of the presumed pattern against a specific threshold. Existing metrics, such as Maximum Independent Set (MIS) and Minimum Node Image (MNI), either demand significant computational time or risk overestimating pattern counts. Our innovative approach aligns with the MIS and identifies independent subgraphs. Through the \qoutes{Maximal Independent Set} metric, this paper offers an efficient solution that minimizes latency and provides users with control over pattern overlap. Through extensive experimentation, our proposed method achieves an average of 10.58x speedup when compared to GraMi and an average 3x speedup when compared to T-FSM.
\end{abstract}

\maketitle
\section{Introduction}~\label{sec:Introduction}
Frequent Subgraph Mining (FSM) has seen significant interest recently due to its applications in various domains such as chemical analysis~\cite{kong:2022:molecule}, bioinformatics~\cite{mrzic:2018:grasping}, social network anomaly detection~\cite{bindu:2016:mining}, and Android malware detection~\cite{martinelli:2013:classifying}. FSM involves identifying recurring subgraphs in a larger graph that exceed a predefined frequency threshold. Despite its growing importance, FSM poses challenges due to its time-consuming and complex nature, especially as it requires identifying high-frequency subgraphs and verifying if they surpass a set threshold - an NP-complete problem. FSM is generally solved in two steps, namely,
\begin{enumerate*}
  \item \textbf{Generation Step}: identifying potential high-frequency subgraphs within the large graph, and
  \item \textbf{Metric Step}: determining if these subgraphs occur more than a predefined threshold.
\end{enumerate*}
Thus, to effectively solve the problem, determining an efficient method for both steps is important.

Considerable research efforts have been directed towards both steps of this approach. Generation Step has been tackled using techniques based on edge extension or vertex extension. In the case of edge extension methods, there are two main approaches. 
\begin{enumerate*}
\item  A possible $k$ edge subgraph is derived from two frequent $(k-1)$-subgraphs. The process of merging two viable subgraphs sharing identical $(k-2)$ edge subgraphs is detailed in multiple works~\cite{kuramochi:2004:efficient, kuramochi:2005:finding}. However, these approaches are either used for transactional graphs, in which the occurrence of patterns in multiple graphs is considered, as opposed to our single graph approach, or the extension is on an edge level. The main disadvantage of this approach is that during the breadth-first search, the candidate pattern graphs with the same number of vertices are placed at different levels hence curbing the early termination criteria based on the number of vertices and support value.
\item The other approach involves the addition of an edge to all frequent $(k-1)$-subgraphs followed by eliminating redundancies~\cite{yuan:2023:tfsm, elseidy:2014:grami}. This approach generates a lot of redundancies and removing them takes a considerable amount of time. 
\end{enumerate*}

In the case of the vertex extension methods as illustrated by works such as~\cite{chen:2020:pangolin, chen:2022:efficient, yao:2020:locality, chen:2021:sandslash}, the potential $k$-vertex subgraphs are established by exhaustively generating all feasible extensions of $(k-1)$-vertex frequent subgraphs, and eliminating redundancies. This method is similar to the second edge extension method and thus suffers from the same drawbacks.


This paper proposes to identify potentially frequent subgraph patterns directly in the pattern space. Our approach generates $k$-vertex candidate patterns by merging two $(k-1)$-vertex patterns that have been already identified as frequent. The implementation of this merging strategy tackles several non-trivial challenges, including the identification of suitable merge points that maintain the graph's meaningful connectivity, handling labels and attributes on vertices and edges to ensure coherent merged patterns, and ensuring that the merged patterns are unique and non-redundant by checking for automorphisms and determining canonical forms.

The Metric Step in Frequent Subgraph Mining has seen advancements, particularly with the \ac{mis} metric. This metric is known for its accuracy and calculates pattern frequency by counting disjoint, independent patterns in the data graph. However, its NP-complete nature requires considerable computational time. An alternative, the \ac{mni}, finds the maximum count of independent sets for a pattern but can result in non-disjoint sets and potential vertex overlap, leading to pattern count overestimation. \cite{yuan:2023:tfsm} introduced a fractional-score method based on \ac{mni}, reducing overestimation by factoring in each vertex's contribution to a pattern. Despite fewer false positives, this method still faces significant inaccuracies~\cite{yuan:2023:tfsm}. A limitation of both MIS and MNI is the lack of user control over the overlap between identified patterns.

In contrast, our proposed approach introduces a mechanism to align with the MIS, which exclusively identifies independent patterns without vertex overlap. This is achieved through the utilization of \ac{mal}. \ac{mal} is an approximation of MIS. We leverage a mathematical relationship that exists between these metrics, which enables us to establish the lower bound of support for the \ac{mal} metric, based on the support for the MIS metric. The upper bound corresponds to the support value for MIS. Users can define the degree of overlap with the MIS metric by selecting the interpolation level, referred to as the \textit{slider value}. This way, they have control over the extent of overlap with the MIS metric, offering a metric applicable across various applications. 
\added{For instance, in network analysis, a degree of overestimation may be acceptable, whereas, in fields like chemical composition analysis or biomedical research, precision is of paramount importance, even if it means potentially missing some patterns. While the overlap chosen by the user may lead to missing some patterns, it ensures that the patterns identified are accurate.}

This paper makes the following contributions,

\begin{itemize}
\item It proposes an innovative method for generating candidate subgraph patterns by merging two frequently occurring smaller subgraph patterns, each with $(k-1)$ vertices, to efficiently form potential larger subgraph patterns with $k$ vertices to effectively prune the search space.
\item The paper introduces a new metric based on the Maximal Independent Set, allowing for the enumeration of pattern graphs within a data graph, with a user-defined slider for controlling the overlap with \ac{mis} metric, enhancing both accuracy and flexibility in frequent subgraph mining.
\item  It conducts extensive experiments to demonstrate the efficiency of the proposed algorithms, showing a significant improvement in computational time and memory usage compared to existing graph mining methods like GraMi and T-FSM.
\end{itemize}


\section{Preliminaries}~\label{sec:preliminaries}

In this section, we introduce the definitions and terminologies for the terms used in the rest of the paper. 

\subsection{General Graph Terminologies}


\subsubsection{Articulation Vertex}
An articulation vertex (AV) in an undirected graph is a vertex whose removal disconnects the graph.
An AV in a directed graph is defined by making all directed edges undirected and then using the 
definition for undirected graphs.

\subsubsection{Graph Isomorphism}
Two labeled graphs $G_1(V_1, E_1)$ and $G_2$ $(V_2, E_2)$, where each vertex $v$ has label $l(v)$, are isomorphic if there exists a bijection $f: V_1 \rightarrow V_2$ such that for any two vertices $u, v \in V_1$, $(u, v) \in E_1 \iff (f(u), f(v)) \in E_2$ and $l(v) = l(f(v))$ for each $v \in V_1$.

\subsubsection{Automorphism}
An automorphism of a graph is a special case of graph isomorphism, where a graph is mapped onto itself. Formally, an automorphism of a labeled graph $G=(V,E)$ is a bijection $f:V \rightarrow V$ such that for any two vertices $u,v \in V$, $(u,v) \in E \iff (f(u),f(v)) \in E$ and $l(v) = l(f(v))$ for each $v \in V$. An automorphism is thus a permutation of $V$.  The set of all automorphisms of a graph can be produced using {\em generators}. Generators are a compact subset of automorphisms whose composition results in the generation of its automorphisms. 

Graph $P_1$ (\cref{fig:patt1}) has two automorphisms: (1) the identity permutation (1,2,3) which maps each vertex onto itself and 
\begingroup
\hbadness=10000 
\begin{figure}[ht]
  \centering
  \begin{subfigure}{0.21\linewidth}
    \centering
    \resizebox{\linewidth}{!}{ 
      \begin{tikzpicture}
  \node[redNode] (u1) at (0, 0) {$u_1$};
  \node[redNode] (u2) at (1.0, 0) {$u_3$};
  \node[blueNode] (u4) at (0.5, 1.0) {$u_2$};
  \draw[doubleArrow] (u1) -- (u2);
  \draw[doubleArrow] (u2) -- (u4);
  \draw[doubleArrow] (u1) -- (u4);
\end{tikzpicture}
    }
    \caption{$P_1$}
    \label{fig:patt1}
  \end{subfigure}
  \begin{subfigure}{0.21\linewidth}
    \centering
    \resizebox{\linewidth}{!}{ 
      \begin{tikzpicture}
  \node[blueNode] (u1) at (0, 1.0) {$u_4$};
  \node[blueNode] (u2) at (1.0, 1.0) {$u_2$};
  \node[redNode] (u3) at (0.5, 0) {$u_1$};
  \draw[doubleArrow] (u1) -- (u2);
  \draw[doubleArrow] (u2) -- (u3);
  \draw[doubleArrow] (u1) -- (u3);
\end{tikzpicture}
    }
    \caption{$P_2$}
    \label{fig:patt2}
  \end{subfigure}
  \begin{subfigure}{0.52\linewidth}
    \centering
    \resizebox{\linewidth}{!}{ 
      \begin{tikzpicture}
  \node[redNode] (d1) at (-1.5, 0) {$d_1$};
  \node[redNode] (d2) at (-0.5, 0) {$d_2$};
  \node[redNode] (d3) at (0.5, 0) {$d_3$};
  \node[redNode] (d4) at (1.5, 0) {$d_4$};

  \node[blueNode] (d5) at (-1.0, 1.0) {$d_5$};
  \node[blueNode] (d6) at (0, 1.0) {$d_6$};
  \node[blueNode] (d7) at (1.0, 1.0) {$d_7$};

  \draw[doubleArrow] (d1) -- (d2);
  \draw[doubleArrow] (d2) -- (d3);
  \draw[doubleArrow] (d3) -- (d4);

  \draw[doubleArrow] (d5) -- (d6);
  \draw[doubleArrow] (d6) -- (d7);

  \draw[doubleArrow] (d1) -- (d5);
  \draw[doubleArrow] (d2) -- (d5);

  \draw[doubleArrow] (d2) -- (d6);
  \draw[doubleArrow] (d3) -- (d6);

  \draw[doubleArrow] (d3) -- (d7);
  \draw[doubleArrow] (d4) -- (d7);
\end{tikzpicture}
    }
    \caption{Data graph $D$}
    \label{fig:datagraph}
  \end{subfigure}
  \caption{The double arrow represents directed edges in both directions. Labels are denoted by vertex colors.}
  \label{fig:example}
    \Description{}
\end{figure}
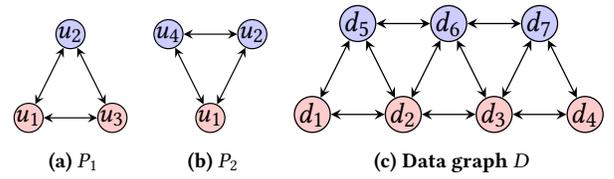
\endgroup
(2) permutation (3,2,1) which maps $u_1$ and $u_3$ to each other and $u_2$ onto itself - note this holds because $u_1$ and $u_3$ have the same color (label).
In this instance, the set of generators is the same as the set of automorphisms. If all vertices in $P_1$ had the same label, it would have six automorphisms corresponding to the 3! permutations of its vertices. These can be generated by using two generators: (2,1,3) which swaps $u_1$ and $u_2$ and (1,3,2) which swaps $u_2$ and $u_3$. 

\subsubsection{Subgraph Isomorphism}
Given vertex-labeled data and pattern graphs $D (V_d,E_d)$ and $P (V_p,E_p)$, respectively,
and a set of labels $L$, where each vertex 
$v$ has a label $l(v) \in L$: $P$ is {\em subgraph isomorphic} 
to (is a subgraph of) $D$ if there exists an injective function 
$f: V_p \rightarrow V_d$ such that (1) $l(p) = l(f(p))$ for all $p \in V_p$ and (2) $(p_1, p_2) \in E_p \implies (f(p_1),f(p_2)) \in E_d$ for all $(p_1, p_2) \in E_p$. {\em Note}: our definition of subgraph isomorphism does not require an edge
in $P$ if one exists between the corresponding vertices in $D$.



\subsubsection{Bliss}

The Bliss library~\cite{JunttilaKaski:ALENEX:2007:bliss,JunttilaKaski:TAPAS:2011:bliss} is used to determine whether two vertex-labeled graphs (undirected or directed)  $G_1$ and $G_2$ are isomorphic. 
It computes (1) a hash (which may not be unique) and (2) 
a canonical form (which is unique) for both $G_1$ and $G_2$. If the hashes are unequal, $G_1$ and $G_2$ are quickly determined to be non-isomorphic. 
If the hashes are equal, their canonical forms are checked and if equal, $G_1$ and $G_2$ are isomorphic. Bliss also computes
generators and automorphisms of a graph.

\subsection{Graph Mining}

Frequent subgraph mining (FSM) identifies pattern graphs that are subgraphs of a given data graph 
$D$ that occur a certain number \acs{support} of times. This paper focuses on FSM~\cite{bringmann:2008:frequent} 
on labeled, directed graphs. However, our techniques can be trivially extended to unlabeled and/or undirected graphs.
Figure~\ref{fig:example} illustrates these concepts. Both patterns $P_1$ (\cref{fig:patt1}) and $P_2$ (\cref{fig:patt2}) are subgraphs of data graph $D$ (\cref{fig:datagraph}). Specifically, $P_1$ can be seen to be a subgraph of $D$
using six functions $f$ that map $(u_1, u_2, u_3)$ to: 
(i) $(d_1,d_5,d_2)$, 
(ii) $(d_2,d_5,d_1)$, 
(iii) $(d_2,d_6,d_3)$,
(iv) $(d_3,d_6,d_2)$,
(v) $(d_3,d_7,d_4)$,
(vi) $(d_4,d_7,d_3)$.

Recall that we approach the graph mining problem in two steps: generation (derive candidate subgraphs) and metric (determine which candidate subgraphs occur in the data graph and are frequent). Using this terminology, both $P_1$ and $P_2$ may be viewed as candidate subgraphs. Although we have not yet formally defined how we count occurrences (i.e., the metric), a threshold of 7 would result in $P_1$ being classified as infrequent in $D$.

\subsubsection{Anti-monotone Property}
 Let $P$ and $Q$ be subgraphs of data graph $D$ such that $P$ is a subgraph of $Q$ and $Q$ is a k-size pattern while $P$ is $(k-1)$-size pattern. Anti-monotone property states the number of instances of $P$ $\geq$ number of instances of $Q$ in $D$.

\subsection{Generation Step terminologies}~\label{sec:def}

\subsubsection{Core Graph and Marked Vertex}~\label{sec:def:coreGraph}
A core graph is obtained from a pattern graph by disconnecting one of its vertices along with its incident edges. $C_m^{u_j}$ denotes the core graph obtained by disconnecting vertex $u_j$ from pattern graph $P_m$. The disconnected vertex $u_j$ is referred to as the marked vertex of the core graph. 
$C_1^{u1}$ denotes the core graph (\cref{fig:core:c1u1}) formed by disconnecting vertex $u_1$ from pattern $P_1$; $u_1$ is the
marked vertex of $C_1^{u1}$.

\begin{figure}[htb]
  \centering
  \begin{subfigure}{0.28\linewidth}
    \centering
    \begin{tikzpicture}
  \node[markedNode] (u1) at (0, 0) {$u_1$};
  \node[blueNode] (u2) at (1, 0) {$u_2$};
  \node[redNode] (u3) at (0, 1) {$u_3$};
  \draw[doubleArrow] (u2) -- (u3);
\end{tikzpicture}
    \caption{CoreGraph $C_1^{u1}$}
    \label{fig:core:c1u1}
  \end{subfigure}
  \begin{subfigure}{0.28\linewidth}
    \centering
    \begin{tikzpicture}
  \node[redNode] (u1) at (0, 0) {$u_1$};
  \node[markedNode] (u2) at (1, 0) {$u_2$};
  \node[redNode] (u3) at (0, 1) {$u_3$};
  \draw[doubleArrow] (u1) -- (u3);
\end{tikzpicture}
    \caption{CoreGraph $C_1^{u2}$}
    \label{fig:core:c1u2}
  \end{subfigure}
  \begin{subfigure}{0.33\linewidth}
    \centering
    \begin{tikzpicture}
  \node[markedNode] (u1) at (0, 0) {$u_1$};
  \node[edgeNode] (u4) at (1.4, 0.80) {$u_4$};
  \node[edgeNode] (u5) at (0.8, 0.25) {$u_5$};
  \node[blueNode] (u2) at (2.0, 0) {$u_2$};
  \node[redNode] (u3) at (0, 1) {$u_3$};
  \draw[singleArrow] (u3) -- (u4);
  \draw[singleArrow] (u4) -- (u2);
  \draw[singleArrow] (u2) -- (u5);
  \draw[singleArrow] (u5) -- (u3);
\end{tikzpicture}
    \caption{ext CoreGraph $E_1^{u1}$}
    \label{fig:ext:core:e1}
  \end{subfigure}
  \caption{CoreGraphs \& extended CoreGraph for pattern $P_1$}
  \label{fig:core:example}
    \Description{}
\end{figure}
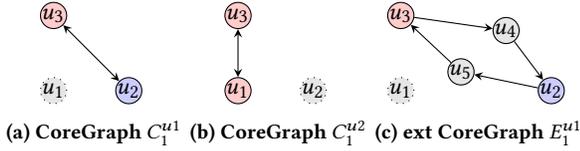


\subsubsection{Core Graph Isomorphism}
Two core graphs are isomorphic iff the core graphs with marked vertices excluded are isomorphic. For example, $C_1^{u1}$ is isomorphic to $C_1^{u3}$. 



\subsubsection{Core Groups} \label{sec:def:coreGroup}

The set of all core graphs generated from patterns of the same size, grouped by isomorphic core graphs, is collectively
referred to as core groups. A single core group is represented as a <key, value> pair,  where the key is a core graph without the marked vertex and the value is a list of all core graphs that are isomorphic to the key and to each other. Three core groups are generated from pattern graphs $P_1$ and $P_2$ (\cref{fig:patt1} and \cref{fig:patt2}): $CoreGroup(C_1^{u3}) = \set{C_1^{u3}, C_1^{u1}, C_2^{u2}, C_2^{u4}}$, $CoreGroup(C_1^{u2}) = \set{C_1^{u2}}$, $CoreGroup(C_2^{u1}) = \set{C_2^{u1}}$.

\subsubsection{Extended Core Graphs}

Although the preceding definitions assume that only vertices are labeled, there exist applications where edges are also labeled. We are able to transform edge-labeled graphs into vertex-labeled graphs as follows: an edge $(u,v)$ with label $L(u,v)$ is replaced with two edges $(u,w)$ and $(w,v)$. The newly introduced vertex $w$ is assigned label $L(u,v)$.
\cref{fig:ext:core:e1} shows the extended CoreGraph of \cref{fig:core:c1u1}, in which $u_4$ is introduced in edge $(u_3, u_2)$ and $u_5$ is introduced in edge $(u_2, u_3)$. 


\subsection{Metric Step}

\begin{figure*}[ht]
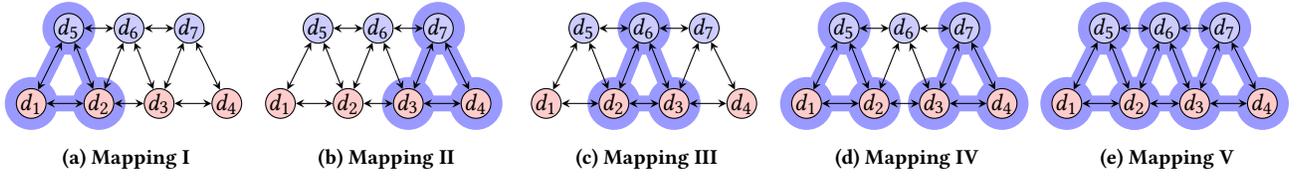

  \centering
  \begin{subfigure}{0.19\linewidth}
    \centering
    \runningDatagraph{5}
    \caption{Mapping I}
    \label{fig:dataGraphMapping1}
  \end{subfigure}
  \begin{subfigure}{0.19\linewidth}
    \centering
    \runningDatagraph{6}
    \caption{Mapping II}
    \label{fig:dataGraphMapping2}
  \end{subfigure}
  \centering
  \begin{subfigure}{0.19\linewidth}
    \centering
    \runningDatagraph{2}
    \caption{Mapping III}
    \label{fig:dataGraphMapping3}
  \end{subfigure}
  \begin{subfigure}{0.19\linewidth}
    \centering
    \runningDatagraph{1}
    \caption{Mapping IV}
    \label{fig:dataGraphMapping4}
  \end{subfigure}
  \begin{subfigure}{0.19\linewidth}
    \centering
    \runningDatagraph{7}
    \caption{Mapping V}
    \label{fig:dataGraphMapping5}
  \end{subfigure}
  \caption{All possible Mappings for pattern $P_1$}
  \label{fig:dataGraphMapping}
\end{figure*}

The metric step counts the number of occurrences of each candidate pattern in the data graph and determines
whether it is frequent. We next review metrics such as MIS, MNI and FS that have been used in the literature 
to determine subgraph frequency counts and introduce our metric, the mIS, a key contribution of this paper.
The \ac{mis} is the gold standard in pattern frequency counting due to its accuracy in providing exact counts. However, its high computational complexity has made it challenging to use. The \ac{mni} serves as a faster, approximate alternative to MIS but can significantly overestimate pattern counts. The Fractional-Score (FS) metric is a variation of MNI that attempts to address the overestimation but still occasionally overestimates certain patterns. To address these drawbacks, we propose \ac{mal}, a metric that preserves the independent set property of MIS while providing the flexibility to choose between overestimation or underestimation of patterns. Additionally, \ac{mal} has the advantage of running with similar runtimes as \ac{mni}, making it a practical choice for subgraph mining tasks. 



Let $M = \set{m_1, m_2, \ldots, m_{|M|}}$ be the set of all possible mappings from $P \to D$. Each mapping $m_i$ represents
an injective function $f: V_p \rightarrow V_d$ represented by a list of of pairs of the form $(v_p, v_d)$ which maps each $v_p \in V_P$ to vertex $v_d = f(v_p) \in D$.

\subsubsection{Independent Set}
An independent set $I \subseteq M$, is a subset of $M$ where for any two distinct mappings $m_i$, $m_j$ in $I$, 
and for any vertex pair $(v_{p_i},v_{d_i}) \in m_i$ and any vertex pair $(v_{p_j}, v_{d_j}) \in m_j$, it holds 
that $v_{d_i} \neq v_{d_j}$. That is there are no two mappings that share the same vertex in $D$. For instance, in the case of pattern $P_1$, an independent set could be any one of the mappings in \cref{fig:dataGraphMapping} except \cref{fig:dataGraphMapping5}.

\subsubsection{\acf{mis}~\cite{kuramochi:2005:finding}}
\ac{mis} is an independent set of maximum size among all independent sets; i.e., there is
no other independent set with a greater number of mappings. For pattern $P_1$, \cref{fig:dataGraphMapping4} is the \ac{mis}.

\subsubsection{\acf{mal}}
\ac{mal} is an independent set such that there is no possible addition of a mapping that preserves independence. For pattern $P_1$, maximal independent sets are depicted in \cref{fig:dataGraphMapping3,fig:dataGraphMapping4}. However, \cref{fig:dataGraphMapping1,fig:dataGraphMapping2} are not \ac{mal}, since additional mappings can be added that
preserve independence.



\subsubsection{\acf{mni}~\cite{elseidy:2014:grami}}
Recall that $M = \set{m_1, \ldots, m_{|M|}}$ 
is the set of all possible mappings from $P \to D$. Let $f_i$ be the injective mapping used in mapping $m_i$. For each vertex $v \in P$, define $F(v) = \set{f_1(v),\ldots,f_{|M|}(v)}$ as the 
set of unique images of $v$ in $D$ under these mappings. \ac{mni} of $P$ in $D$ is $min \set{|F(v)|:v \in V_p} $. 
For $P_1$ in $D$, $F(u_1) = F(u_3) = \set{d_1, d_2, d_3, d_4}$, $F(u_2) = \set{d_5, d_6, d_7}$, giving an MNI of 3
(\cref{fig:dataGraphMapping5}). 




\changed{
\subsubsection{Fractional-Score~\cite{yuan:2023:tfsm}}
A Fractional-Score based method is an extension of \ac{mni}, which uses fractional scores to reduce overestimation. The fractional score is determined based on the contribution of each vertex to its neighbors. For instance, in \cref{fig:example}, for data graph $D$ and pattern graph, $P_1$ the contribution of $d_5$ to $d_1$ is $1/2$, this is because there are two vertices with the same label that $d_5$ can equally contribute to. With this value, the total contribution of each vertex to a specific label is calculated to be $3$. However, this value is still more than the \ac{mis} value of 2. Thus, the fractional score method still overestimates the support value.
}


\section{FLEXIS}~\label{sec:flexis}

The \ac{flexis} framework presents an approach for discovering frequent subgraphs within a data graph. 

\subsection{Methodological Contributions}

\subsubsection{Contribution 1}~\label{sec:contri1}
The previous section introduced a new metric \ac{mal} which adopts the best features of \ac{mis} and \ac{mni}. It crucially retains the independence property (i.e., that vertices do not overlap) of \ac{mis} and thus overcomes the overestimation problem of \ac{mni}~\cite{yuan:2023:tfsm}. Like \ac{mni}, \ac{mal} can be computed quickly, whereas computing \ac{mis} is expensive due to its NP-hardness. In addition, \ac{mal} provides a user-controlled parameter $\lambda$ that can be used to tune the accuracy-speed trade-off of the algorithm. The theoretical basis for this is the important approximation result included below.
\begin{theorem}~\label{theo:mal}
Given a pattern graph with $n$ vertices. Let $m$ denote the number of mappings in a {\bf maximal} (mIS) independent set and 
$M$ the number of mappings in a {\bf maximum} independent set. Then $m \le M \le mn$.
\end{theorem}
\begin{proof}
$m \le M$ follows immediately from the definitions of mIS and MIS. If the second inequality is false, we have $M > mn$.
The $m$ mappings of the mIS together comprise $mn$ distinct vertices in the data graph due to the independence property 
of mIS. Each of these $mn$ vertices from the mIS can be allocated to 
a different mapping of MIS. However, since $M > mn$, it follows that there is at least one mapping $\mu$ in MIS whose data graph 
vertices do not include any of the $mn$ mIS vertices. Then, $\mu$ can be added to the mIS, contradicting that it was maximal.
\end{proof}
We show that the bound in \cref{theo:mal} is tight using the example in \cref{fig:malVisual} with pattern $P_m$ containing $n~=~4$ vertices.  An MIS computation yields $M = 4$ mappings shown in blue, green, yellow, and orange. A \ac{mal} computation either picks the same 4 mappings as MIS or picks the single mapping (i.e., $m = 1$) shown in red. The latter scenario gives $mn = M = 4$, yielding a tight bound.
\begin{figure}[htb]
    \centering
    \begin{subfigure}{0.66\linewidth}
      \centering
      \begin{tikzpicture}

  \def\smallerScale{18}
  \def\biggerScale{22}
  \def\smallerWidth{4}
  \def\biggerWidth{6}

  \node[redNode] (d2) {$d_2$};
  \node[redNode] (d1) [above of=d2] {$d_1$};
  \node[redNode] (d3) [below right of=d2] {$d_3$};
  \node[redNode] (d4) [below left of=d2] {$d_4$};
  \node[redNode] (d5) [right of=d1] {$d_5$};
  \node[redNode] (d6) [above right of=d5] {$d_6$};
  \node[redNode] (d7) [below right of=d5] {$d_7$};
  \node[redNode] (d8) at (160:1.1) {$d_8$};
  \node[redNode] (d9) [above of=d8] {$d_9$};
  \node[redNode] (d10) [left of=d8] {$d_{10}$};
  \node[redNode] (d11) [below right of=d3] {$d_{11}$};
  \node[redNode] (d12) [right of=d11] {$d_{12}$};
  \node[redNode] (d13) [below of=d11] {$d_{13}$};
  \node[redNode] (d14) [below left of=d4] {$d_{14}$};
  \node[redNode] (d15) [below of=d14] {$d_{15}$};
  \node[redNode] (d16) [left of=d14] {$d_{16}$};

  \draw[singleEdge] (d2) -- (d1);
  \draw[singleEdge] (d2) -- (d3);
  \draw[singleEdge] (d2) -- (d4);
  \draw[singleEdge] (d5) -- (d1);
  \draw[singleEdge] (d5) -- (d6);
  \draw[singleEdge] (d5) -- (d7);
  \draw[singleEdge] (d8) -- (d2);
  \draw[singleEdge] (d8) -- (d9);
  \draw[singleEdge] (d8) -- (d10);
  \draw[singleEdge] (d11) -- (d3);
  \draw[singleEdge] (d11) -- (d12);
  \draw[singleEdge] (d11) -- (d13);
  \draw[singleEdge] (d14) -- (d4);
  \draw[singleEdge] (d14) -- (d15);
  \draw[singleEdge] (d14) -- (d16);

  \begin{scope}[on background layer]
    \node[circularNode, fill=blue!40, scale=\biggerScale] at (d1) {};
    \node[circularNode, fill=blue!40, scale=\biggerScale] at (d5) {};
    \node[circularNode, fill=blue!40, scale=\biggerScale] at (d6) {};
    \node[circularNode, fill=blue!40, scale=\biggerScale] at (d7) {};
    \draw[singleEdge, fill=blue!40, draw=blue!40, line width=\biggerWidth] (d5) -- (d1);
    \draw[singleEdge, fill=blue!40, draw=blue!40, line width=\biggerWidth] (d5) -- (d6);
    \draw[singleEdge, fill=blue!40, draw=blue!40, line width=\biggerWidth] (d5) -- (d7);

    \node[circularNode, fill=green!40, scale=\biggerScale] at (d2) {};
    \node[circularNode, fill=green!40, scale=\biggerScale] at (d8) {};
    \node[circularNode, fill=green!40, scale=\biggerScale] at (d9) {};
    \node[circularNode, fill=green!40, scale=\biggerScale] at (d10) {};
    \draw[singleEdge, fill=green!40, draw=green!40, line width=\biggerWidth] (d8) -- (d2);
    \draw[singleEdge, fill=green!40, draw=green!40, line width=\biggerWidth] (d8) -- (d9);
    \draw[singleEdge, fill=green!40, draw=green!40, line width=\biggerWidth] (d8) -- (d10);

    \node[circularNode, fill=yellow!40, scale=\biggerScale] at (d3) {};
    \node[circularNode, fill=yellow!40, scale=\biggerScale] at (d11) {};
    \node[circularNode, fill=yellow!40, scale=\biggerScale] at (d12) {};
    \node[circularNode, fill=yellow!40, scale=\biggerScale] at (d13) {};
    \draw[singleEdge, fill=yellow!40, draw=yellow!40, line width=\biggerWidth] (d11) -- (d3);
    \draw[singleEdge, fill=yellow!40, draw=yellow!40, line width=\biggerWidth] (d11) -- (d12);
    \draw[singleEdge, fill=yellow!40, draw=yellow!40, line width=\biggerWidth] (d11) -- (d13);

    \node[circularNode, fill=orange!40, scale=\biggerScale] at (d4) {};
    \node[circularNode, fill=orange!40, scale=\biggerScale] at (d14) {};
    \node[circularNode, fill=orange!40, scale=\biggerScale] at (d15) {};
    \node[circularNode, fill=orange!40, scale=\biggerScale] at (d16) {};
    \draw[singleEdge, fill=orange!40, draw=orange!40, line width=\biggerWidth] (d14) -- (d4);
    \draw[singleEdge, fill=orange!40, draw=orange!40, line width=\biggerWidth] (d14) -- (d15);
    \draw[singleEdge, fill=orange!40, draw=orange!40, line width=\biggerWidth] (d14) -- (d16);

    \node[circularNode, fill=red!40, scale=\smallerScale] at (d2) {};
    \node[circularNode, fill=red!40, scale=\smallerScale] at (d1) {};
    \node[circularNode, fill=red!40, scale=\smallerScale] at (d3) {};
    \node[circularNode, fill=red!40, scale=\smallerScale] at (d4) {};
    \draw[singleEdge, fill=red!40, draw=red!40, line width=\smallerWidth] (d2) -- (d1);
    \draw[singleEdge, fill=red!40, draw=red!40, line width=\smallerWidth] (d2) -- (d3);
    \draw[singleEdge, fill=red!40, draw=red!40, line width=\smallerWidth] (d2) -- (d4);
  \end{scope}

\end{tikzpicture}
      \caption{Data graph $D_m$}
    \end{subfigure}
    \begin{subfigure}{0.32\linewidth}
      \centering
      \begin{tikzpicture}
  \node[redNode] (u2) {$u_2$};
  \node[redNode] (u1) [above of=u2] {$u_1$};
  \node[redNode] (u3) [below right of=u2] {$u_3$};
  \node[redNode] (u4) [below left of=u2] {$u_4$};

  \draw[singleEdge] (u1) -- (u2);
  \draw[singleEdge] (u2) -- (u3);
  \draw[singleEdge] (u2) -- (u4);
\end{tikzpicture}
      \caption{Pattern graph $P_m$}
    \end{subfigure}
    \caption{Maximal independent set visualization}
    \label{fig:malVisual}
    \Description{}
    red: $\set{d_1,d_2,d_3,d_4}$, blue: $\set{d_1,d_5,d_6,d_7}$,
    green: $\set{d_2,d_8,d_9,d_{10}}$, yellow: $\set{d_3,d_{11},d_{12},d_{13}}$, orange: $\set{d_4,d_{14},d_{15},d_{16}}$
\end{figure}
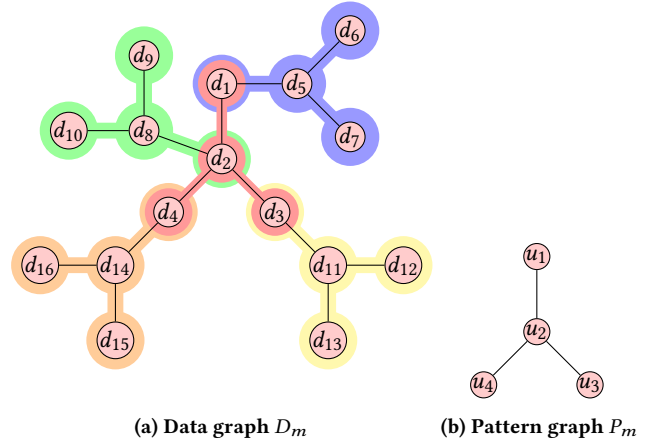

In classical FSM, patterns that occur at least $\sigma$ (the support) times in the data graph are considered frequent. We
instead consider a pattern to be frequent if it occurs at least $\tau$ times. 
\begin{equation}
\tau = \lfloor \sigma (1 - 1/n) \lambda + \sigma/n \rfloor
\label{eqn:tau}
\end{equation}
where $n$ denotes the number of vertices in the pattern and the user-defined parameter $\lambda$ is in the range $[0, 1]$. Notice that when $\lambda = 1$, $\tau = \sigma$ and when $\lambda = 0$, $\tau = \sigma /n$. The choice of $\sigma/ n$ as one of the endpoints results from the second inequality of Theorem~\ref{theo:mal}, which shows that \ac{mal} approximates \ac{mis} to within a factor of $n$. Relative to the ``gold standard" \ac{mis} metric, $\lambda = 1$ guarantees that there will be no false positives while $\lambda = 0$ guarantees there will be no false negatives. A $\lambda$ in between trades off 
between these two scenarios. We discuss these with some examples below.

The number of occurrences of $P_1$ in $D$ (Figure~\ref{fig:example}) depends on the
metric: e.g., \ac{mis} gives 2 (\cref{fig:dataGraphMapping4}), while \ac{mni} gives 3 (\cref{fig:dataGraphMapping5}).
If run to completion, \ac{mal} gives either 1 (\cref{fig:dataGraphMapping3}) or 2 (\cref{fig:dataGraphMapping4}).
Our implementation of \ac{mal} allows us to terminate when a given $\tau$ is reached. Thus, if $\tau =1$, \ac{mal} could
return the match shown in \cref{fig:dataGraphMapping1} or \cref{fig:dataGraphMapping2} in addition to
\cref{fig:dataGraphMapping3}.
Now, assume $\lambda = 1$ so that $\tau = \sigma$. If $\sigma = 3$, $P_1$ is not considered frequent under \ac{mis} and \ac{mal}, but
is frequent under \ac{mni} (which allows overlapping vertices). This example illustrates the additional pruning that 
occurs when using \ac{mal} relative to \ac{mni}.
If $\lambda = 1$ as before, but $\sigma = 2$ giving $\tau =2$, \ac{flexis} gives either \cref{fig:dataGraphMapping3} and declares $P_1$ infrequent or \cref{fig:dataGraphMapping4} and correctly declares $P_1$ frequent. However, if 
$\sigma = 2$ and $\lambda = 0.25$, $\tau = 1$ for a size-3 pattern. Now, \ac{flexis} terminates after finding \cref{fig:dataGraphMapping1}, \cref{fig:dataGraphMapping2}, or \cref{fig:dataGraphMapping3}. Here, \ac{flexis} 
correctly determines that $P_1$ is frequent. 

When $\lambda = 0$, \ac{flexis} overestimates the number of 
frequent patterns. This approximates the behavior of $\ac{mni}$, which also overestimates the number of frequent 
patterns, but never reports an infrequent pattern. However, early termination allows \ac{flexis} to run faster than 
\ac{mis}.

\added{
This versatility grants users the flexibility to fine-tune the trade-off between accuracy and processing time, tailoring the metric to suit the specific needs of their applications. It is important to emphasize that our approach achieves performance levels akin to both \ac{mni} and MIS but at significantly reduced computational time, maintaining nearly equivalent accuracy when compared to traditional methods.
}

\subsubsection{Contribution 2}
The concept of merging two $k-1$-size frequent patterns to find possible $k$-size frequent patterns has been utilized
previously~\cite{kuramochi:2005:finding}. In that work, the size of the pattern graph was defined as the number of edges, 
whereas we define size as the number of vertices. (A graph with 3 vertices and 2 edges thus has size 2 in their method and 
size 3 in ours.) This distinction is important because the vertex-based approach prunes the search space more effectively, resulting in better performance. 

\ac{flexis} starts by identifying all size-2 patterns (edges) in the data graph. It then uses a matcher (vf3matcher) to
determine which size-2 patterns are frequent. These are merged to create candidate patterns of size 3. The candidate
patterns are evaluated to identify the frequent ones, which are then merged, etc. We illustrate this using size-3 patterns $P_1$ and $P_2$ (from Figure~\ref{fig:example}) and merging them to form size-4 candidate pattern graphs (\cref{fig:mergeProcessExamplev2}). 

\begin{figure}[ht]
  \centering
  \begin{tikzpicture}

  

  \node[redNode] (u1) at (0, 0) {};
  \node[redNode] (u2) at (1.0, 0) {};
  \node[blueNode] (u4) at (0.5, 1.0) {};
  \draw[doubleArrow] (u1) -- (u2);
  \draw[doubleArrow] (u2) -- (u4);
  \draw[doubleArrow] (u1) -- (u4);

  \node[blueNode] (u1) at (0, -0.75) {};
  \node[blueNode] (u2) at (1.0, -0.75) {};
  \node[redNode] (u3) at (0.5, -1.75) {};
  \draw[doubleArrow] (u1) -- (u2);
  \draw[doubleArrow] (u2) -- (u3);
  \draw[doubleArrow] (u1) -- (u3);


  \begin{scope}[shift={(3.0, 1.0)}]
    \node[redNode, dotted] (u1) at (0.0, 0.0) {};
    \node[redNode] at (-0.5, -0.5) {};
    \node[blueNode] at (+0.5, -0.5) {};
    \draw[doubleArrow] (-0.35, -0.5) -- (+0.35, -0.5);
  \end{scope}

  \begin{scope}[shift={(3.0, -0.25)}]
    \node[redNode, dotted] (ud) at (0.0, 0.0) {};
    \node[blueNode] at (-0.5, -0.5) {};
    \node[blueNode] at (+0.5, -0.5) {};
    \draw[doubleArrow] (-0.35, -0.5) -- (+0.35, -0.5);
  \end{scope}

  \begin{scope}[shift={(3.0, -1.5)}]
    \node[blueNode, dotted] (ud) at (0.0, 0.0) {};
    \node[redNode] at (-0.5, -0.5) {};
    \node[blueNode] at (+0.5, -0.5) {};
    \draw[doubleArrow] (-0.35, -0.5) -- (+0.35, -0.5);
  \end{scope}

  
  \begin{scope}[shift={(5.5, 1.0)}]
    \node[redNode] (u1) at (-0.5, 0.0) {};
    \node[blueNode] (u2) at (+0.5, 0.0) {};
    \node[redNode] (u4) at (0.0, +0.5) {};
    \node[redNode] (u3) at (0.0, -0.5) {};
    \draw[doubleArrow] (u1) -- (u2);
    \draw[doubleArrow] (u1) -- (u3);
    \draw[doubleArrow] (u1) -- (u4);
    \draw[doubleArrow] (u2) -- (u3);
    \draw[doubleArrow] (u2) -- (u4);
  \end{scope}

  \begin{scope}[shift={(5.5, -0.5)}]
    \node[redNode] (u1) at (-0.5, 0.0) {};
    \node[blueNode] (u2) at (+0.5, 0.0) {};
    \node[redNode] (u4) at (0.0, +0.5) {};
    \node[blueNode] (u3) at (0.0, -0.5) {};
    \draw[doubleArrow] (u1) -- (u3);
    \draw[doubleArrow] (u1) -- (u4);
    \draw[doubleArrow] (u2) -- (u3);
    \draw[doubleArrow] (u2) -- (u4);
    \draw[doubleArrow] (u1) -- (u2);
  \end{scope}

  \begin{scope}[shift={(5.5, -2.0)}]
    \node[redNode] (u1) at (-0.5, 0.0) {};
    \node[blueNode] (u2) at (+0.5, 0.0) {};
    \node[redNode] (u4) at (0.0, +0.5) {};
    \node[blueNode] (u3) at (0.0, -0.5) {};
    \draw[doubleArrow] (u1) -- (u2);
    \draw[doubleArrow] (u1) -- (u3);
    \draw[doubleArrow] (u1) -- (u4);
    \draw[doubleArrow] (u2) -- (u3);
    \draw[doubleArrow] (u2) -- (u4);
    \draw[doubleArrow] (u3) -- (u4);
  \end{scope}

  \coordinate (a) at (1.2, 0.75);
  \coordinate (b) at (2.3, 1.75);
  \coordinate (c1) at (1.4, 0.05);
  \coordinate (d1) at (2.3, 0.50);
  \coordinate (c) at (1.5, -0.75);
  \coordinate (d) at (2.3, -0.65);
  \coordinate (e) at (1.2, -1.35);
  \coordinate (f) at (2.0, -1.75);

  \coordinate (g) at (3.8, 0.60);
  \coordinate (h) at (4.8, 0.85);
  \coordinate (i) at (3.7, -1.70);
  \coordinate (j) at (4.8, -0.5);
  \coordinate (i2) at (3.8, -0.65);
  \coordinate (j2) at (4.8, -1.65);
  \coordinate (k1) at (3.8, 0.50);
  \coordinate (l1) at (4.8, -0.35);
  \coordinate (k2) at (3.8, 1.55);
  \coordinate (l2) at (4.8, 1.15);
  \coordinate (k) at (3.8, 0.35);
  \coordinate (l) at (4.9, -1.45);
  \coordinate (k3) at (3.8, -1.9);
  \coordinate (l3) at (4.8, -1.9);




  \draw[flowArrow] (c1) -- (d1);
  \draw[flowArrow] (c) -- (d);
  \draw[flowArrow] (e) -- (f);
  
  \draw[flowArrow] (g) -- (h) node[midway, sloped, above, text=red] {$x2$};
  \draw[flowArrow] (i) -- (j);
  \draw[flowArrow] (i2) -- (j2);
  \draw[flowArrow] (k) -- (l);
  \draw[flowArrow] (k1) -- (l1);
  \draw[flowArrow] (k3) -- (l3);



  
  \node at (0.5, -3.0) {(k-1)-size pattern};
  \node at (3.0, -3.0) {CoreGraphs};
  \node at (5.5, -3.0) {k-size pattern};

\end{tikzpicture}
  \caption{Merge process}
  \label{fig:mergeProcessExamplev2}
    \Description{}
\end{figure}
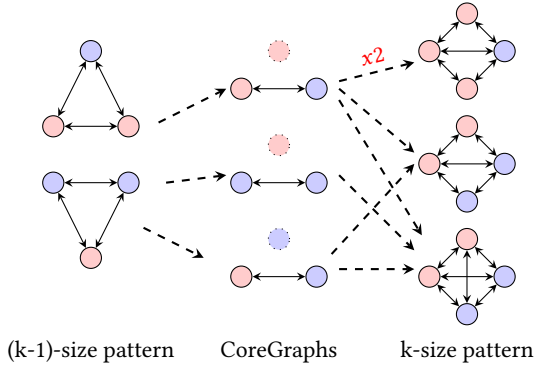

To gain intuition about why this approach prunes more efficiently relative to edge-based merging, consider a data graph 
with $40$ vertices and $\tau = 10$. No frequent pattern can have more than $40/10 = 4$ vertices because \ac{mal} does not permit overlap. Patterns of size 5 or more do not need to be considered and can be pruned. \ac{flexis} will not require more than three merging steps to generate candidates ($2 \rightarrow 3$, $3 \rightarrow 4$, $4 \rightarrow 5$). Edge-based merging~\cite{kuramochi:2005:finding} requires more merging steps and incurs more runtime. 

Other methods such as GraMi~\cite{elseidy:2014:grami} and T-FSM~\cite{yuan:2023:tfsm} do not use merging. but instead,  
extend frequent patterns by adding an edge to generate larger patterns. Merging two frequent graphs generates fewer candidates, making \ac{flexis} more efficient than GraMi and T-FSM.

\subsubsection{Summary}
Our method has the following advantages,
\begin{enumerate*}
  \item \ac{mal} is closely related to the actual frequent pattern count metric \ac{mis}, while providing comparable or better speedup than \ac{mni}.
  \item \ac{mal} is the first metric that allows the user to control the accuracy-speed trade-off.
  \item \ac{flexis} is the first method that uses the merging of two frequent patterns based on vertices rather than edges, resulting in a faster search process
\end{enumerate*}


\subsection{Algorithmic flow of FLEXIS}~\label{sec:algorithm}


  \ac{flexis} outlined in \cref{alg:mining} consists of the following steps:

  \begin{enumerate}
      \item Candidate Pattern Generation (Generation Step): Initially, \ac{flexis} generates all size-2 candidate patterns (i.e., edges) ($CP$) from the data graph $G$ (\cref{alg:mining:edges}).
      \item Finding Frequent Patterns (Metric Step): All frequent patterns ($FP$) within the candidate set $CP$ are found by utilizing a modified version of Vf3Light~\cite{carletti:2018:vf3,carletti:2019:vf3}, with our custom metric \ac{mal} integrated, to ensure non-overlapping vertices (\cref{alg:mining:matcher}). Note that the original implementation of Vf3Light uses the \ac{mni} metric.
      \item Candidate Pattern Expansion (Next Generation Step): Subsequently, \ac{flexis} merges frequent patterns ($FP$) to generate candidate patterns ($CP$) that are one vertex larger (\cref{alg:mining:merger}).
      \item Steps $2$ and $3$  continue until all frequent patterns are exhausted (\cref{alg:mining:loop}).
  \end{enumerate}

\begin{algorithm}
  \small
  \caption{\ac{flexis}}
  \label{alg:mining}
  \begin{algorithmic}[1]
  \Require {$G, \sigma, \lambda$}
  \Ensure {freqPattList}
  \Procedure{mining}{$G, \sigma$}
  \State {$CP \gets \FunctionCall{edges}{G}$} ~\label{alg:mining:edges}
  \State {$\text{freqPattList} \gets \emptyset$}
  \sWhile{$\Not{} CP.\FunctionCall{empty}{}$ } {~\label{alg:mining:loop}
    \State determine $\tau$ from $CP, \sigma, \lambda$ using Eqn~\ref{eqn:tau}~\label{line:mining:lerp}
    \State $FP \gets \FunctionCall{vf3Matcher}{G, CP, \tau}$~\label{alg:mining:matcher}
    \State $\text{freqPattList} \gets \text{freqPattList} \cup FP$
    \State $CP \gets \FunctionCall{generateNewPatterns}{FP}$~\label{alg:mining:merger}
  }
  \State \Return freqPattList
  \EndProcedure
\end{algorithmic}

\end{algorithm}

We explain in detail our approach in the Generation and the Metric Steps, in the following sections. 

\subsubsection{Generation Step}

\Call{generateNewPatterns}{} (\cref{alg:generate:New:Patterns:nonclique}) combines $(k-1)$-vertex frequent patterns to form $k$-vertex candidate patterns. \cref{alg:gen:nc:cg:gen} of the function first computes core groups (\cref{sec:def}) from the input set $P^{(k-1)}$ consisting of all $(k-1)$-vertex frequent  patterns. Each core group identified by $cgID$ is processed in \cref{alg:gen:nc:cg:iter1} and each pair of core graphs $C_1$ and $C_2$ in the core group $cgID$ is considered in \cref{alg:gen:nc:cg:iter2}. 
The next step is to merge pairs of core graphs $C_1$ and $C_2$ (\cref{alg:gen:nc:ma}): recall that all core graphs (excluding their marked vertices) 
in a core group are isomorphic; each core graph can be denoted by a common $(k-2)$-vertex component $\Gamma$ and a disconnected marked vertex. Core graphs 
differ from each other in that their marked vertices may have different labels and may be attached to the vertices in $\Gamma$ in different ways. The merge step adds the marked vertices from $C_1$ and $C_2$ to $\Gamma$ to get a $k$-vertex pattern. However, to generate all possible $k$-vertex candidate patterns, 
we must explicitly consider every automorphism of $\Gamma$ (\cref{alg:gen:nc:fa,alg:gen:nc:pa}). This process generates all non-clique candidate patterns. A modification to this process (merging in a third graph) generates all clique candidate patterns (\cref{alg:gen:nc:gnc}) and adds them to the set of candidate $P^k$ patterns. After generating all the $P_k$, we use \Call{RemoveDuplicates}{} to remove duplicate patterns in the candidate generation using Bliss's canonical form (\cref{alg:line:removeDuplicates}).
\begin{algorithm}
  \small   
  \caption{Candidate Pattern Generation}
  \label{alg:generate:New:Patterns:nonclique}


\begin{algorithmic}[1]
  \Require {Set of size-$(k-1)$ pattern graphs $P^{(k-1)}$}
  \Ensure {Set of size-$k$ pattern graphs $P^{k}$}
  \Function{generateNewPatterns}{$ P^{(k-1)} $}
  \State  {\em cgroups} $\gets \FunctionCall{CoreGroup}{P^{(k-1)}} $ \label{alg:gen:nc:cg:gen}
    \State $ P^k \gets \emptyset $
    \sForAll{ $cgID \in$ {\em cgroups} }{~\label{alg:gen:nc:cg:iter1}
      \sForAll{$[C_i, C_j] \in \Call{pairs}{cgID}$}{~\label{alg:gen:nc:cg:iter2}
        \State $\text{autoList} \gets \Call{findAutomorphism}{cgID, C_i, C_j)}$~\label{alg:gen:nc:fa}
        \sForAll{ $\alpha \gets \text{autoList} $ } {~\label{alg:gen:nc:pa}
          \State $\text{P} \gets \Call{merge}{C_i, C_j, \alpha} $~\label{alg:gen:nc:ma}
          \State $P^k.\Call{add}{P}$
          \State $P^k \gets P^k \cup \Call{generateCliques}{P, C_i, C_j, cgroups}$~\label{alg:gen:nc:gnc}
        }
      }
    }
    \State $\Call{RemoveDuplicates}{P^k}$~\label{alg:line:removeDuplicates}
    \State \Return $P^k$
  \EndFunction
\end{algorithmic}
\end{algorithm}

\cref{fig:mergeProcessExamplev2} depicts \Call{generateNewPatterns}{}, the process of taking 3-vertex patterns, decomposing them into core graphs, and then merging them to form 4-vertex patterns. Non-clique patterns (such as the top two 4-vertex patterns in \cref{fig:mergeProcessExamplev2}) are formed by merging two core graphs, whereas clique patterns (such as the bottom 4-vertex pattern in \cref{fig:mergeProcessExamplev2}) are obtained by additionally merging in a third core graph. The \Call{merge}{} operation (Line 8) is a key step within \Call{generateNewPatterns}{}. We illustrate it below with examples:
\paragraph{Example 1 (simple merging)}
\cref{fig:merge:patternResultant} shows three examples that use core graphs derived from patterns $P_1$ and $P_2$ (\cref{fig:patt1,fig:patt2}). In \cref{fig:merged:pattern:c1u1c1u1}, core graph $C_1^{u_1}$ is merged with itself. Here, two copies of the marked vertex $u_1$ are reattached to
the core consisting of vertices $u_2$ and $u_3$. \cref{fig:merged:pattern:c1u2c1u2} similarly shows $C_1^{u_2}$ being merged with itself. Two copies 
of $u_2$ are reattached to the core ($u_1$ and $u_3$). Finally, \cref{fig:mergePatt3} shows the result of merging $C_1^{u_3}$ and $C_2^{u_4}$. 
\begin{figure}[htb]
  \centering
  \begin{subfigure}{0.30\linewidth}
    \centering
    \begin{tikzpicture}
  \node[redNode] (u1) at (0, 0) {$u_1$};
  \node[blueNode] (u2) at (1, 0) {$u_2$};
  \node[redNode] (u3) at (0, 1) {$u_3$};
  \node[redNode] (u4) at (1, 1) {$u_1$};
  \draw[doubleArrow] (u1) -- (u2);
  \draw[doubleArrow] (u2) -- (u3);
  \draw[doubleArrow] (u1) -- (u3);
  \draw[doubleArrow] (u3) -- (u4);
  \draw[doubleArrow] (u2) -- (u4);
\end{tikzpicture}
    \caption{$C_1^{u_1} \cup C_1^{u_1}$}
    \label{fig:merged:pattern:c1u1c1u1}
  \end{subfigure}
  \begin{subfigure}{0.30\linewidth}
    \centering
    \begin{tikzpicture}
  \node[redNode] (u1) at (0, 0) {$u_1$};
  \node[blueNode] (u2) at (1, 0) {$u_2$};
  \node[redNode] (u3) at (0, 1) {$u_3$};
  \node[blueNode] (u4) at (1, 1) {$u_2$};
  \draw[doubleArrow] (u1) -- (u2);
  \draw[doubleArrow] (u2) -- (u3);
  \draw[doubleArrow] (u1) -- (u3);
  \draw[doubleArrow] (u3) -- (u4);
  \draw[doubleArrow] (u1) -- (u4);
\end{tikzpicture}
    \caption{$C_1^{u_2} \cup C_1^{u_2}$}
    \label{fig:merged:pattern:c1u2c1u2}
  \end{subfigure}
  \begin{subfigure}{0.30\linewidth}
    \centering
    \begin{tikzpicture}
  \node[redNode] (u1) at (0, 0) {$u_1$};
  \node[blueNode] (u2) at (1, 0) {$u_2$};
  \node[redNode] (u3) at (0, 1) {$u_3$};
  \node[blueNode] (u4) at (1,1) {$u_4$};
  \draw[doubleArrow] (u1) -- (u2);
  \draw[doubleArrow] (u2) -- (u3);
  \draw[doubleArrow] (u1) -- (u3);
  \draw[doubleArrow] (u2) -- (u4);
  \draw[doubleArrow] (u1) -- (u4);
\end{tikzpicture}
    \caption{$C_1^{u_3} \cup C_2^{u_4}$}
    \label{fig:mergePatt3}
  \end{subfigure}
  
  \caption{$\cup$ refers to merging to corresponding CoreGraphs}
  \label{fig:merge:patternResultant}
    \Description{}
\end{figure}
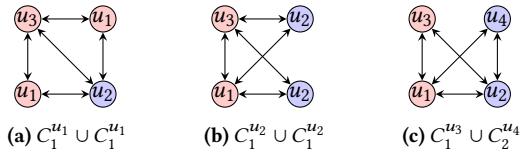

\paragraph{Example 2 (merging with automorphisms)}~\label{sec:gen:perm}
Next, we illustrate the more complicated scenario of merging with automorphisms in \cref{fig:trailedTriangle}. \cref{fig:tailTri} shows the 
4-vertex pattern $P$. We will merge core graph $C^{u_4}$ (obtained from $P$ by disconnecting $u_4$ from $u_2$) with itself. 
\cref{fig:mergeTri1} shows the simple merge where two copies of $u_4$ are reattached to $u_2$. 

However, $\Gamma$, the triangle obtained by deleting $u_4$ from $P$ has an automorphism $\alpha$ that maps $u_2$ and $u_3$ to each other 
(because both have red labels) and $u_1$ to itself. This gives a second merged graph (\cref{fig:mergeTri1}) obtained by reattaching $u_4$ 
to $u_2$ and attaching a copy of $u_4$ to $u_3 = \alpha(u_2)$ as a result of the automorphism (we call this Case 1). If $u_4$ had instead originally 
been connected to $u_1$ (the sole blue vertex in $P$), the automorphism would not be used since it maps $u_1$ to itself (we call this Case 2). 
\Call{findAutomorphism}{} (called in Line 6 of \Call{generateNewPatterns}{} and shown in \cref{alg:core:findautomorphism}) handles these
intricacies. A key point is that \Call{findAutomorphism}{} does not merely return automorphisms of $\Gamma$. Instead, it additionally 
examines the nature of the connection of the marked vertex to $\Gamma$ using \Call{differentTopology}{} to determine whether we are in
Case 1 or Case 2. 

\begin{figure}[htb]
    \centering
  \begin{subfigure}{0.30\linewidth}
    \centering
    \begin{tikzpicture}
  \node[blueNode] (u1) at (0, 0) {$u_1$};
  \node[redNode] (u2) at (1, 0) {$u_2$};
  \node[redNode] (u3) at (0, 1) {$u_3$};
  \node[redNode] (u4) at (1, 1) {$u_4$};
  \draw[doubleArrow] (u1) -- (u2);
  \draw[doubleArrow] (u2) -- (u3);
  \draw[doubleArrow] (u1) -- (u3);
  \draw[doubleArrow] (u2) -- (u4);
\end{tikzpicture}
    \caption{$P$}
    \label{fig:tailTri}
  \end{subfigure}
  \begin{subfigure}{0.30\linewidth}
    \centering
        \begin{tikzpicture}
      \node[blueNode] (u1) at (0, 0) {$u_1$};
      \node[redNode] (u2) at (1, 0) {$u_2$};
      \node[redNode] (u3) at (0, 1) {$u_3$};
      \node[redNode] (u4) at (0.8, 1) {$u_4$};
      \node[redNode] (u5) at (1.4, 1) {$u_4$};
      \draw[doubleArrow] (u1) -- (u2);
      \draw[doubleArrow] (u2) -- (u3);
      \draw[doubleArrow] (u1) -- (u3);
      \draw[doubleArrow] (u2) -- (u4);
      \draw[doubleArrow] (u2) -- (u5);
    \end{tikzpicture}
    \caption{$C^{u_4} \cup C^{u_4}$}
    \label{fig:mergeTri1}
  \end{subfigure}
  \begin{subfigure}{0.30\linewidth}
    \centering
    \begin{tikzpicture}
  \node[blueNode] (u1) at (0, 0) {$u_1$};
  \node[redNode] (u2) at (1, 0) {$u_2$};
  \node[redNode] (u3) at (0, 1) {$u_3$};
  \node[redNode] (u4) at (0.85, 1.0) {$u_4$};
  \node[redNode] (u5) at (1.5, 0.8) {$u_4$};
  \draw[doubleArrow] (u1) -- (u2);
  \draw[doubleArrow] (u2) -- (u3);
  \draw[doubleArrow] (u1) -- (u3);
  \draw[doubleArrow] (u3) -- (u4);
  \draw[doubleArrow] (u2) -- (u5);
\end{tikzpicture}
    \caption{$C^{u_4} \cup \alpha(C^{u_4})$}
    \label{fig:mergeTri2}
  \end{subfigure}
    \caption{Pattern graphs with $k-1$ and $k$ size}
    \label{fig:trailedTriangle}
    \Description{}
\end{figure}
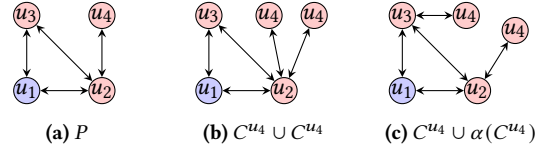

 \begin{algorithm}[ht]
   \caption{checks if merged in different configs possible}
   \label{alg:core:findautomorphism}
   \begin{algorithmic}[1]
  \Require{CoreGraphs $\set{C_i, C_j}$}
  \Ensure{Possible Permutations (perms)}
  \Function{findAutomorphism}{$icg, \set{C_i, C_j}$}
    \State $generators \gets \Call{Gen}{icg}$ \label{alg:fa:gen}
    \sForAll { $cg \in \set{C_i, C_j}$ } {
      \sForAll { $gen \in generators$ } {
        \State $\text{store marked vertex neighbors } (topo)$
        \sIf{$\Call{differentTopology}{cg, gen, topo}$} { \label{alg:fa:dt}
          \State \Return \textit{set of all automorphisms}
        }
      }
    }
  \State \Return \textit{identity permutation}
  \EndFunction
\end{algorithmic}
 \end{algorithm}

We next provide the lemmas and theorem used in the Generation Step.

\begin{lemma} \label{lemma:dia}
  Let $G$ be a connected graph with more than $2$ vertices. Let $u$ and $v$ be $2$ vertices of $G$ whose distance from each other is equal to the diameter of $G$ ($d(u, v) = dia(G)$). Then  $u$ and $v$ are \ac{nav} of $G$ ~\cite{kingan:2022:GraphsAndNetworks,Clark:2005:AFirstLookAtGraphTheory}.
\end{lemma}

\begin{lemma} \label{lemma:nAV}
Let $G$ be a connected non-clique graph with at least three vertices. $G$ has at least two \ac{nav} that are not joined by an edge.
\end{lemma}
\begin{proof}
 Since $G$ is not a clique, $dia(G) \ge 2$. Consider two vertices $u$ and $v$ in $G$ such that $d(u, v) = dia(G)$. From Lemma~\ref{lemma:dia}, $u$ and $v$ are \ac{nav}s. Since $d(u, v) > 1$, they are not joined by an edge. Thus $u, v$ satisfy the lemma.
\end{proof}

\begin{lemma} \label{lemma:candidateGeneration}
Every connected, non-clique $k$-vertex pattern graph can be generated by merging two connected $(k-1)$-vertex pattern graphs
(i.e., by the \Call{Merge}{} function).
\end{lemma}

\begin{proof}
Consider any non-clique k-vertex candidate pattern graph $P$. From \cref{lemma:nAV}, $P$ has two non-articulation vertices $u$ and $v$. The two $(k-1)$-vertex graphs $P_u$ and $P_v$, formed respectively by separately removing $u$ and $v$ from $P$ are connected. The process of forming core groups from $(k-1)$-vertex patterns in \added{\Call{generateNewPatterns}{}(\cref{alg:gen:nc:cg:gen})} removes $v$ from $P_u$ and $u$ from $P_v$, giving the common core graph (i.e., $P$ with both $u$ and $v$ removed). The \Call{Merge}{} function, essentially the reverse of this process, constructs $P$ by starting with the common core graph and adding $u$ and $v$ with the appropriate edges. 
\end{proof}

\textbf{Clique generation}: We next extend the idea of merging $(k-1)$-vertex pattern graphs to
form $k$-vertex cliques. As before, we begin by merging two $(k-1)$-vertex pattern graphs. If both 
are cliques, there is a possibility that they are both sub-patterns of a $k$-vertex clique. Note
that the pattern generated by merging two cliques using \Call{Merge}{} is not itself a clique (the two 
marked vertices are not joined by an edge). Therefore, we need to find and merge a suitable third $(k-1)$-vertex 
pattern graph that will supply the missing edge. 

Figure~\ref{fig:mergeCliqueExample} shows an example.
\begin{figure}[h]
  \begin{tikzpicture}
    \runningDistinctPatterns{1}{0}{0}{0}{$P_3$}   
    \runningDistinctPatterns{2}{0}{2}{0}{$P_4$}   
    \runningDistinctPatterns{3}{0}{4}{0}{$P_5$}   
    \runningDistinctPatterns{4}{0}{6}{0}{$P_6$}   

    \node[align=center] at (3, -0.75) {Pattern Graphs}; 

    \runningDistinctPatterns{1}{3}{0.00}{-3.25}{$C_3^3$} 
    \runningDistinctPatterns{2}{3}{0.00}{-5.25}{$C_4^3$} 

    \runningDistinctPatterns{5}{0}{3.00}{-2.25}{$C_3^3 \cup C_4^3$} 
    \runningDistinctPatterns{1}{2}{3.00}{-4.25}{$C_3^2$} 
    \runningDistinctPatterns{3}{3}{3.00}{-6.25}{$C_5^3$} 

    \runningDistinctPatterns{6}{0}{6.00}{-4.25}{$C_3^3 \cup C_4^3 \cup C_5^3$} 

    \draw[flowArrow, shorten <=10pt, shorten >=22pt] (0.5, -2.75) -- (3.5, -1.75); 
    \draw[flowArrow, shorten <=15pt, shorten >=30pt] (0.5, -4.75) -- (3.5, -1.75); 

    \draw[flowArrow, shorten <=30pt, shorten >=30pt] (3.5, -1.75) -- (6.5, -3.75); 
    \draw[flowArrow, shorten <=15pt, shorten >=30pt] (3.5, -5.75) -- (6.5, -3.75); 

  \end{tikzpicture}
  \caption{Merging Clique Example}~\label{fig:mergeCliqueExample}
  \Description{Merging Clique Example}
\end{figure}
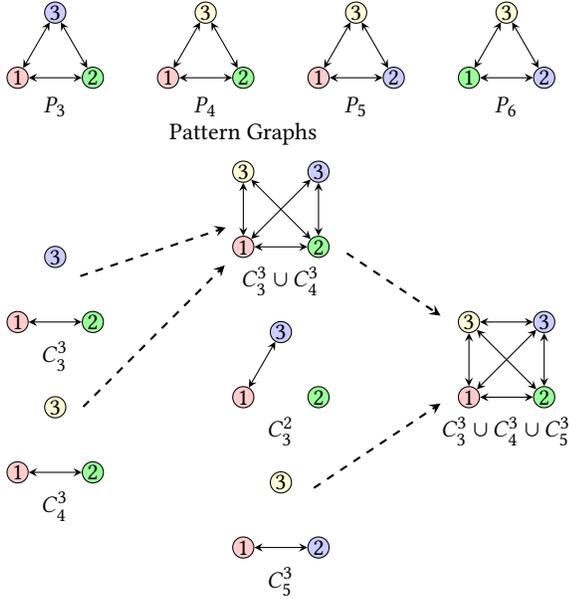
We start by merging patterns $P_3$ and $P_4$ (note that both are 3-cliques). Core graphs $C_3^{3}$ (obtained from $P_3$) and 
$C_4^{3}$ (obtained from $P_4$) are merged to form $M_1$ ($C_3^{3} \cup C_4^{3}$). $M_1$ has the potential to form a 
clique by adding an edge between the two marked vertices. We will now identify a third pattern to merge with $M_1$ to complete 
the clique. An exhaustive search to find such a pattern can be very expensive. We streamline the search by focusing on another core graph $C_3^{2}$ obtained from pattern $P_3$. The idea is that merging two core graphs along one of the 4-clique's four edges and then a third graph along a diagonal can potentially complete the missing edge of the clique. 
All core graphs in the core group with ID $C_3^{2}$ and a marked vertex label matching that of $C_4^{3}$ (i.e., a yellow vertex) are examined to locate the missing edge. In the example, $C_5^{3}$ meets these criteria. Merging with $M_1$ gives the clique. Should this core graph be absent, the pattern does not qualify as a candidate clique. 

A final post-processing step confirms that all other sub-cliques of size 3 are frequent. In our example, this would mean checking that pattern $P_6$ is frequent.

That is in general,
\begin{itemize}
  \item Consider two (k-1) cliques whose core graphs have a common (k-2)-subclique in addition to the one marked vertex each.
  \item After merging the two core graphs, we get a k-clique minus the edge between the two marked vertices ($M$).
  \item To generate the clique, we need to find a third (k-1) clique, whose (k-2)-subclique core graph is different from the first two, but comes from the same pattern as one of the first two. Also, the marked vertex of the third core graph must have the same label as the marked vertex of the second core graph.
  \item If such a CoreGraph is found, we merge it with $M$ to generate the clique.
  \item In the post-processing step, we check if all the (k-1)-size sub-graphs of the generated clique are frequent. If yes, the generated clique is set as a candidate pattern. Otherwise, the generated clique is discarded.
\end{itemize}


The generation of the clique pattern is an extension of the non-clique pattern generation (\cref{alg:generate:New:Patterns:nonclique}). We present the algorithm for generating Clique, as shown in \cref{alg:generateNewCliques}. The algorithm first starts by checking if the frequent pattern graphs are cliques, as this must be true for the candidate pattern graph to be a clique (trivially true, checked in \cref{alg:line:cl:cliques}). Then the algorithm strategically searches through the CoreGroups to determine if the missing Clique, which contains the missing edge of the candidate clique is present (\cref{alg:line:cl:otherCores,alg:line:cl:anoCliqs,alg:line:cl:anoSameLabel}). Upon finding the required Clique, the algorithm merges the two cliques to form the candidate clique pattern (\cref{alg:line:cl:merge}). This pattern is then added to the candidate pattern set (\cref{alg:line:cl:add}).


\begin{algorithm}
  \caption{algorithm for generating candidate cliques}
  \label{alg:generateNewCliques}
  \begin{algorithmic}[1]
    \Require {intermediates from algorithm~\ref{alg:generate:New:Patterns:nonclique}}
    \Ensure {Cliques with size $k$}
    \Function{generateCliques}{$P, C_i, C_j, cgroups$}
        \sIf{ $\Call{pattern}{C_i} \OR \Call{pattern}{C_j} \neq $ cliques } {~\label{alg:line:cl:cliques}
            \State \Return $\emptyset$
        }
        \State $P_l \gets \{\}$
        \State $c3s \gets \Call{OtherCoresSamePattern}{C_i, cgroups}$~\label{alg:line:cl:otherCores}
        \sForAll { $C_k \in c3s$ } {~\label{alg:line:cl:anoCliqs}
          \sIf { $\Call{SameMarkedLabel}{C_j, C_k}$ } {~\label{alg:line:cl:anoSameLabel}
              \State $pc \gets \Call{mergeClique}{P, C_k}$~\label{alg:line:cl:merge}
              \sIf { $pc$ }{
                \State $P_l \gets P_l \cup pc$~\label{alg:line:cl:add}
              }
          }
        }
        \State \Return $P_l$
    \EndFunction
\end{algorithmic}

\end{algorithm}

\begin{lemma} \label{lemma:clique:candidateGeneration}
  Every clique $k$-size candidate pattern graph can be generated by using three $(k-1)$-size frequent clique patterns.
\end{lemma}

\begin{proof}

Removing a single edge, $e = {v_a, v_b}$, from $G_k$ (a clique of $k$-size) results in a non-clique graph, $G'_k$, with vertices ${v_a, v_b}$ as non-articulation vertices. Then according to Lemma~\ref{lemma:candidateGeneration} every non-clique $k$-size pattern graph can be generated by merging two $(k-1)$-size frequent patterns. Thus, $G'_k$ can be reconstructed from two $(k-1)$-size patterns. Since in $G'_k$, every vertex is interconnected except for $v_a$ and $v_b$, the two $(k-1)$-size patterns used to form this non-clique must themselves be cliques. 

To reintroduce $e$ and restore the complete clique structure of $G_k$, we seek a third $(k-1)$-size pattern that includes both $v_a$ and $v_b$. This pattern must share a $(k-2)$-size subgraph (an anchor) with $G'_k$, ensuring compatibility for the merging process. By definition, merging requires one of $e$’s vertices to be the marked vertex, and the other to be part of the anchor. By the same reasoning as above this subgraph will also be a clique.

Thus, three (k-1)-size clique patterns are required to form a k-size clique.
\end{proof}

\begin{theorem} \label{theo:candidateGeneration}
Every candidate pattern can be generated by \added{\ac{flexis}}.
\end{theorem}
\begin{proof}
The proof follows by induction on the number of vertices $k$. For $k=2$ (the base case), we began by computing all frequent edges in the data graph. In the inductive step, we use \cref{lemma:candidateGeneration,lemma:clique:candidateGeneration} to argue that all $k$-patterns can be 
generated by merging two or three $(k-1)$-patterns. 

At each stage, prior to merging, candidate patterns are discarded if they are not frequent.
A $k$ pattern can only be frequent if all of its $(k-1)$- sub-patterns are frequent. 
Therefore, discarding infrequent sub-patterns does not eliminate any frequent patterns.    
\end{proof}

\subsubsection{Metric/Matcher Step}~\label{sec:matcher_step}
\added{We check the frequency of the candidate patterns using \ac{mal} metric (\cref{sec:contri1}). A candidate pattern is deemed frequent based on a user defined parameter ($\lambda$) and support parameter ($\sigma$) (\cref{eqn:tau}). We now explain the modifications to Vf3Light for using this metric.}

\paragraph{Algorithmic Implementation of Matcher Step}

Vf3Light~\cite{carletti:2018:vf3, carletti:2019:vf3}, is a subgraph matching algorithm that is used to find all the embeddings of a pattern 
graph in the data graph. Vf3Light is a fast and efficient depth-first search-based algorithm that uses matching order and properties of graph isomorphism for early termination and effective pruning. VF3Light finds embeddings that may have automorphisms and overlapping vertices across embeddings. Both of these properties are not desirable for \ac{mal}. Thus, we modify VF3Light \changed{\Call{Vf3LightM}{}}) to find \ac{mal} of data vertices given the pattern graph. 
\added{
\ac{flexis} incorporates the following modifications to VF3Light (this description assumes familiarity with VF3Light):
\begin{enumerate}
  \item Pruning: Once the embedding succeeds, none of the vertices from that embedding can be used again (independent set violation). So, we modified the matching engine to stop the search for new embeddings with the same root vertex if a match is found. This is done by setting a flag 
  if a match is found. The flag is reset when the root vertex is unstacked from the recursion stack
  \item Independent Set: To make sure the embeddings do not have overlapping vertices, we use a bitmap to store the used vertices. The bitmap is set when a vertex is used in an embedding. The bitmap is shared across different instances of VFState for the same pattern graph to avoid expensive copies.
\end{enumerate}
}

\begin{algorithm}
  \small
  \caption{VF3Matcher}
  \label{alg:matching}
  \begin{algorithmic}[1]
  \Require {$G, CP, \tau$}
  \Ensure {freqPatterns $f$}
  \Procedure{vf3Matcher}{$G, CP, \tau$}
  \State $f \gets \emptyset$~\label{line:vf3matcher:fempty}
  \sForAll { $x \in CP$ } {~\label{line:vf3matcher:for}
    \sIf{ $\Call{Vf3LightM}{G, x, \tau}$ }{~\label{line:vf3matcher:isfrequent}
      $f \gets f \cup x$~\label{line:vf3matcher:fintersect}
    }
  }
  \State \Return $f$~\label{line:vf3matcher:return}
  \EndProcedure
\end{algorithmic}

\end{algorithm}

\subsection{Complexity Discussion}

A formal asymptotic analysis of \ac{flexis} is not meaningful because of the inherent complexity of the 
FSM problem and its subproblems: graph isomorphism and 
subgraph isomorphism, which are not known to be 
in $P$. Concretely, the runtime of the algorithm depends on
three factors: (1) the number of iterations \#iter of the 
while loop (line 4, Algorithm 1), (2) the run-time $R_{gen}$ of the candidate generation step (\Call{generateNewpatterns}{}), and (3) the runtime $R_{metric}$ of the metric step (\Call{VF3Matcher}{}). \#iter depends inversely on the threshold $\tau$ (the lower its value, the greater the
number of candidate patterns and frequent patterns at
each iteration and the greater is \#iter).
$R_{gen}$ depends on Bliss performance to solve the graph isomorphism ($GI$) problem multiple times. There is presently no 
polynomial-time algorithm for $GI$, nor a proof of NP-completeness. Indeed, this is an active area of research
in theoretical Computer Science. The 
Bliss runtime is not polynomially bounded and depends on characteristics of the input graph,
including size and symmetry. Finally, $R_{metric}$
depends on counting the number of instances of a pattern graph in a data graph. This is related to subgraph matching,
which is NP-complete. \ac{flexis} uses VF3Light, which is a backtracking method whose worst-case complexity is exponential. 

\section{Results}~\label{sec:results}

In this section, we evaluate \ac{flexis} and compare it with GraMi and T-FSM. For this, we utilize the code provided by the authors for GraMi and T-FSM. The original version of GraMi uses Java 17.0.7 and T-FSM uses C++. We implement our method in C++ compiled by g++\-11. We ran our experiments on Xeon CPU E5-4610 v2 (2.30GHz) on Ubuntu 22.04. Our system has 192 GB DDR3 memory.

\textbf{Datasets}
We use $5$ real-world graph datasets with varied numbers of vertices, edges and labels. The datasets included are as follows, 
\begin{enumerate}
    \item Gnutella08~\cite{leskovec:2007:graph, ripeanu:2002:mapping} constitutes a specific instance captured from the Gnutella peer-to-peer file sharing network.
    \item The soc-Epinions1~\cite{richardson:2003:trust} is derived from the online social network Epinions.com, a consumer review website. 
    \item The Slashdot~\cite{leskovec:2009:community} reveals user interconnections on Slashdot's tech news platform. 
    \item The wiki-vote dataset~\cite{leskovec:2010:signed, leskovec:2010:predicting} is a compilation of voting interactions within Wikipedia.
    \item MiCo~\cite{elseidy:2014:grami} dataset is Microsoft co-authorship information.
\end{enumerate}

Table \ref{tab:datasets} shows the properties of these data graphs. $|V_d|$ refers to the maximum degree of the graph, $|V_l|$ is the number of distinct vertex labels, and $|E_l|$ is the number of distinct edge labels. Vertex and edge labels are randomly assigned to vertices and edges, respectively. 

\begin{table}
  \centering
  \caption{Datasets used for experiments}
  \label{tab:datasets}
  \begin{tabular}{|c|c|c|c|c|c|}
  \toprule
  datagraph & $|V|$  & $|V_d|$ & $|V_l|$ & $|E|$ & $|E_l|$ \\
  \midrule
  Gnutella & 6301 & 48 & 5 & 20777 & 5 \\ \hline
  Epinions & 75879 & 1801 & 5 & 508837 & 5 \\ \hline
  Slashdot & 82168 & 2511 & 5 & 948464 & 5 \\ \hline
  wiki-Vote & 7115 & 893 & 5 & 103689 & 5 \\ \hline
  MiCo & 100000 & 21 & 29 & 1080298 & 106 \\
  \bottomrule
\end{tabular}

\end{table}

\textbf{Metrics}
The support threshold \acs{support}, in conjunction with a given data graph, and the user defined parameter $\lambda$, serves as the input for the proposed system. The assessment of our system encompasses the execution time incurred during the mining process, analysis of the frequent pattern obtained, and the quantification of memory utilization. We compare our method to both the modifications of the state-of-art algorithm GraMi, namely the AGrami,  where $\alpha$ is set to $1$, and the baseline GraMi version. We also compare our algorithm with the newest graph mining system T-FSM. In T-FSM we compare our method with both the \ac{mni} version and the fractional score version. Conversely, for our developed algorithm denoted as \ac{flexis}, we present two distinct variants. The first variant entails an implementation absent of any approximation mechanisms. In contrast, the second variant pursues the generation of maximal frequent patterns, a subset of which intersects optimally with those identified by GraMi and T-FSM. Since the problem can run for a long time in the case of large graphs, we restrict the time that both the algorithms run to 30 minutes, after which the system times out. If an algorithm times out then the corresponding values will be missing in the graphs. 

\subsection{Experimental Setup}

Our method uses an undirected data loader and a directed matching algorithm. Thus, to determine the efficiency of performance, only looking at the undirected version might not be fair. Since GraMi has both undirected and directed versions we choose to consider the directed version of GraMi. Moreover, T-FSM proposes a better undirected version of GraMi with some optimizations. Thus, instead of considering the original undirected GraMi, we use the T-FSM version of GraMi, referred to as T-FSM-MNI. We also compare our algorithm with the T-FSM fractional score model (referred in paper as T-FSM). It is important to note that our observations between GraMi and T-FSM do not follow directly from their paper. This is because we consider directed GraMi, while their paper considered the undirected version. Also, since the paper itself postulated the equivalent performance of their \ac{mni} version with that of GraMi and better performance of their fractional-score method when compared to GraMi, we do not reinforce this point again in this paper, rather here we use their important observations.

\begin{figure}
    \centering
    \includegraphics[width=0.6\linewidth]{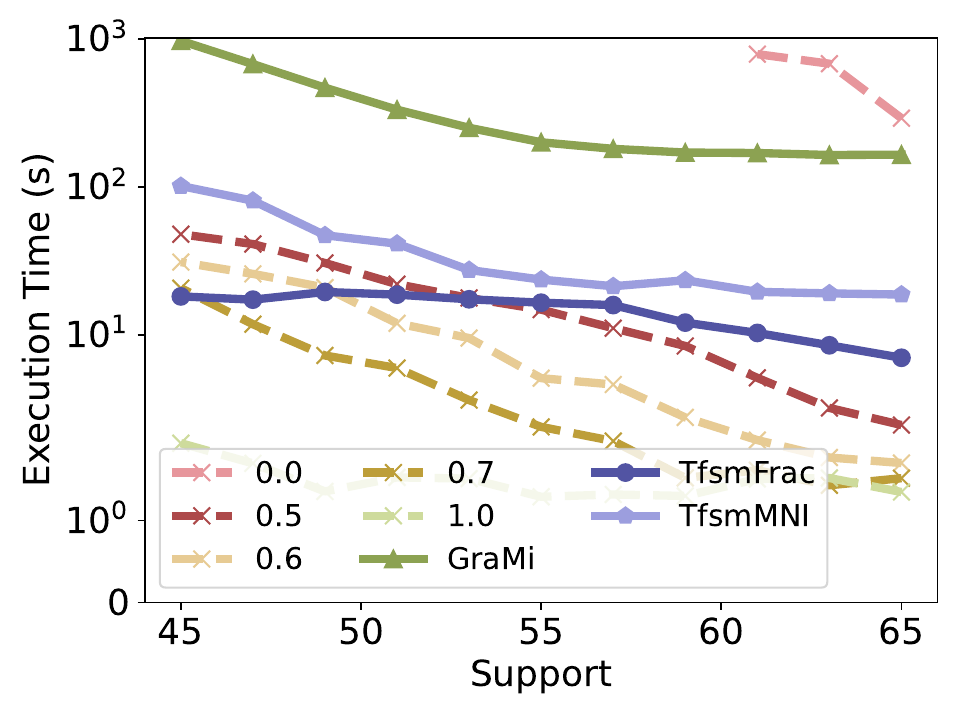}
    \caption{Results on Gnutella}
    \label{fig:log:TimeTaken:Gnutella}
    \vspace{-10pt}
    \Description{}
\end{figure}

\begin{figure*}[htb]
  \centering
  \begin{subfigure}{0.245\linewidth}
    \includegraphics[width=1.0\linewidth]{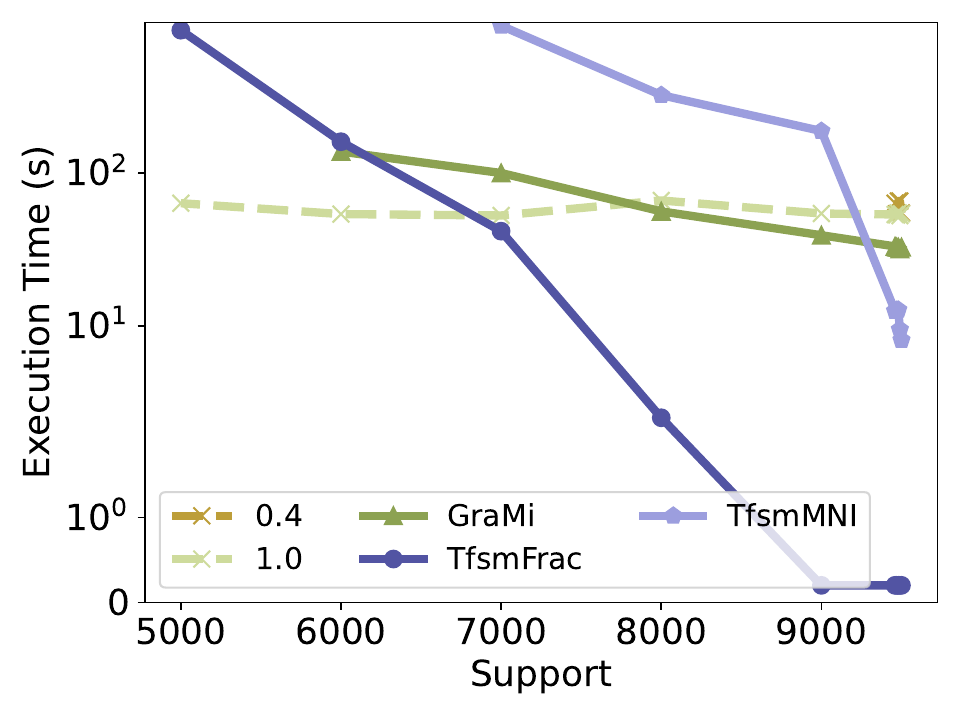}
    \caption{MiCo}
  \end{subfigure}
  \begin{subfigure}{0.245\linewidth}
    \includegraphics[width=1.0\linewidth]{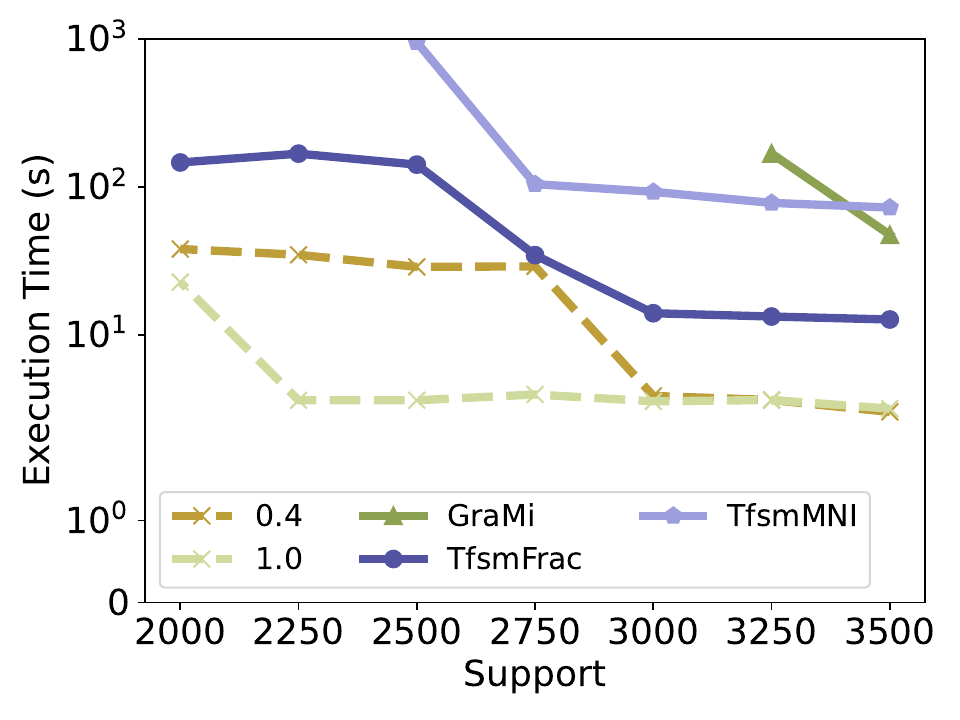}
    \caption{Slashdot}
  \end{subfigure}
  \begin{subfigure}{0.245\linewidth}
    \includegraphics[width=1.0\linewidth]{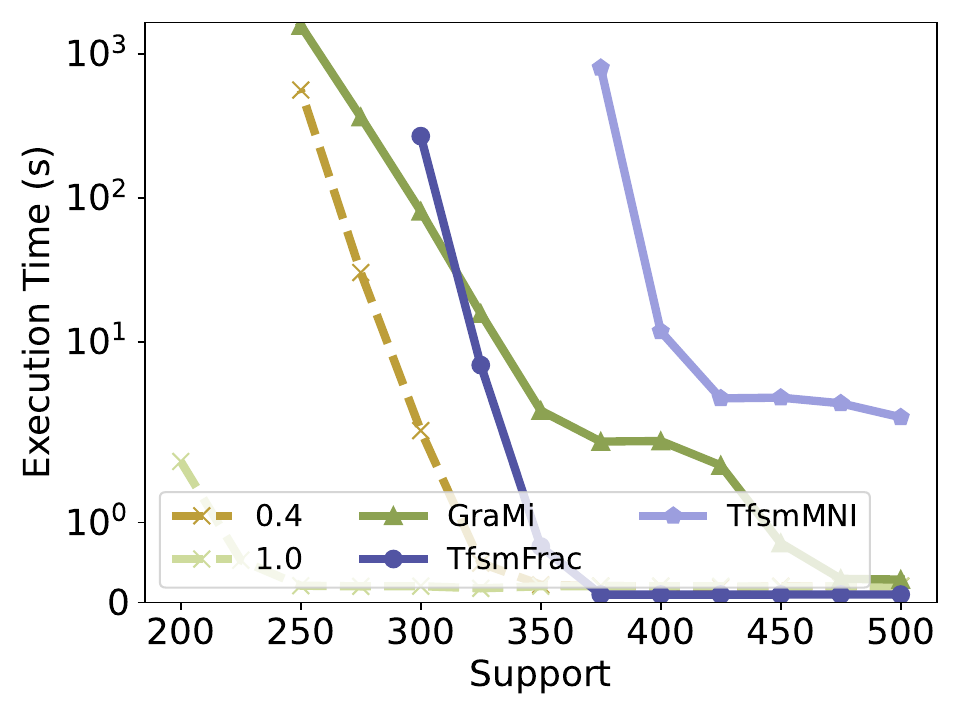}
    \caption{wiki-Vote}
  \end{subfigure}
  \begin{subfigure}{0.245\linewidth}
    \includegraphics[width=1.0\linewidth]{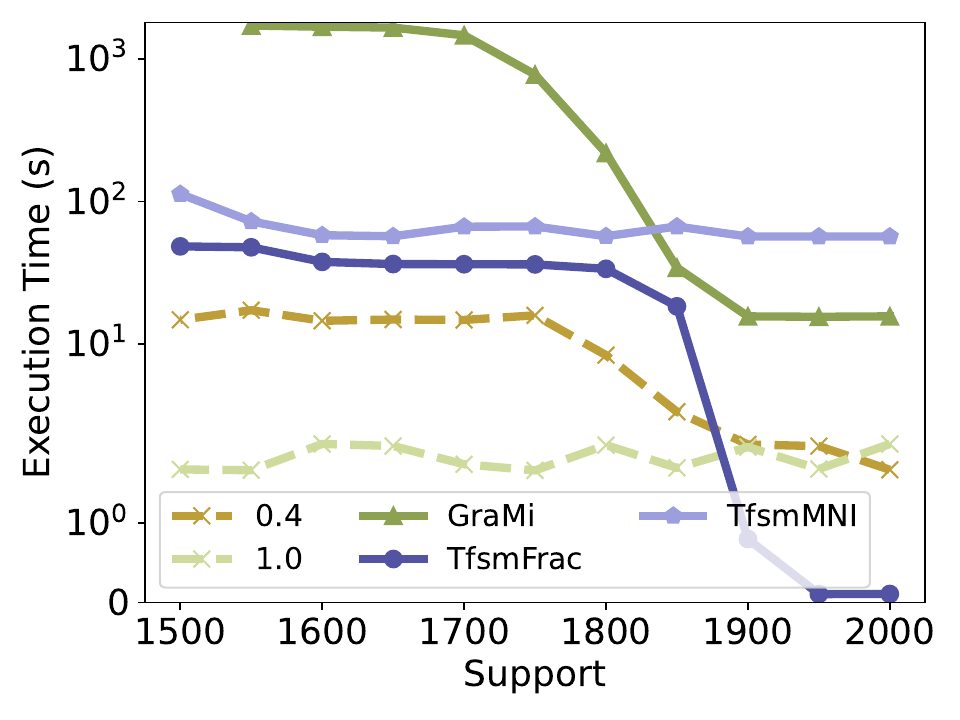}
    \caption{Epinions}
  \end{subfigure}

  \caption{Execution Time}
  \label{fig:log:TimeTaken}
  \Description{}
\end{figure*}

\begin{figure*}[htb]
  \centering
  \begin{subfigure}{0.245\linewidth}
    \includegraphics[width=1.0\linewidth]{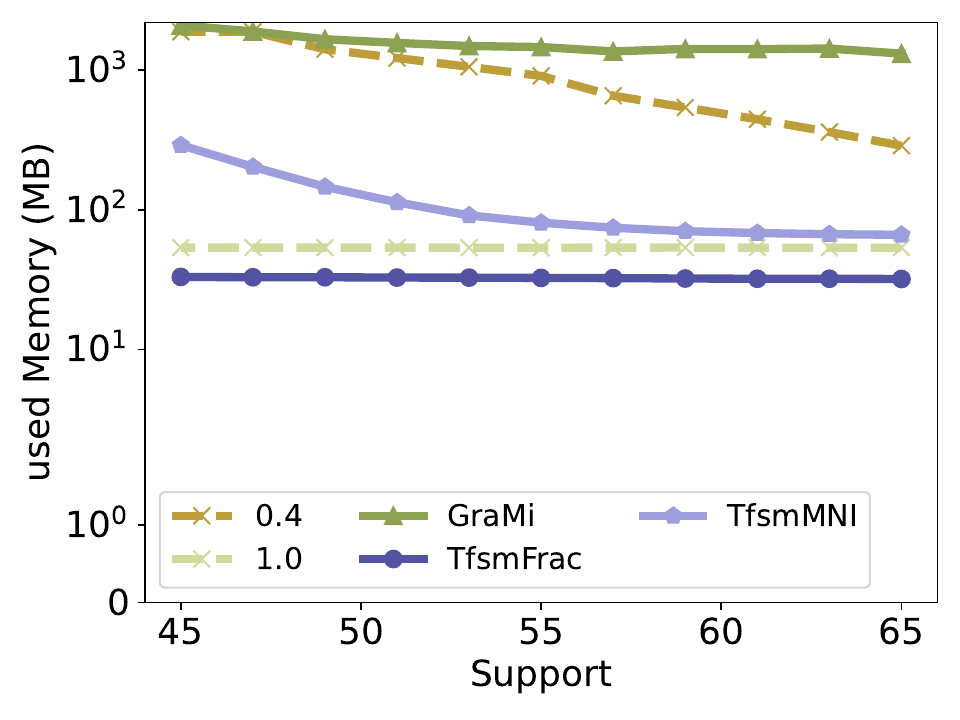}
    \caption{Gnutella}
  \end{subfigure}
  \begin{subfigure}{0.245\linewidth}
    \includegraphics[width=1.0\linewidth]{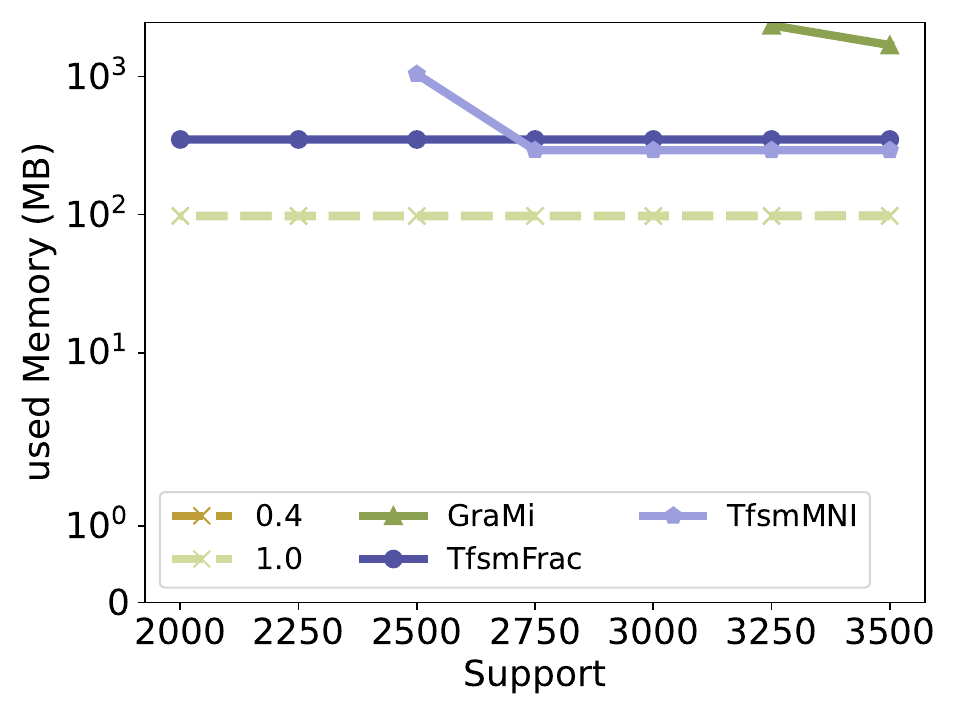}
    \caption{Slashdot}
  \end{subfigure}
  \begin{subfigure}{0.245\linewidth}
    \includegraphics[width=1.0\linewidth]{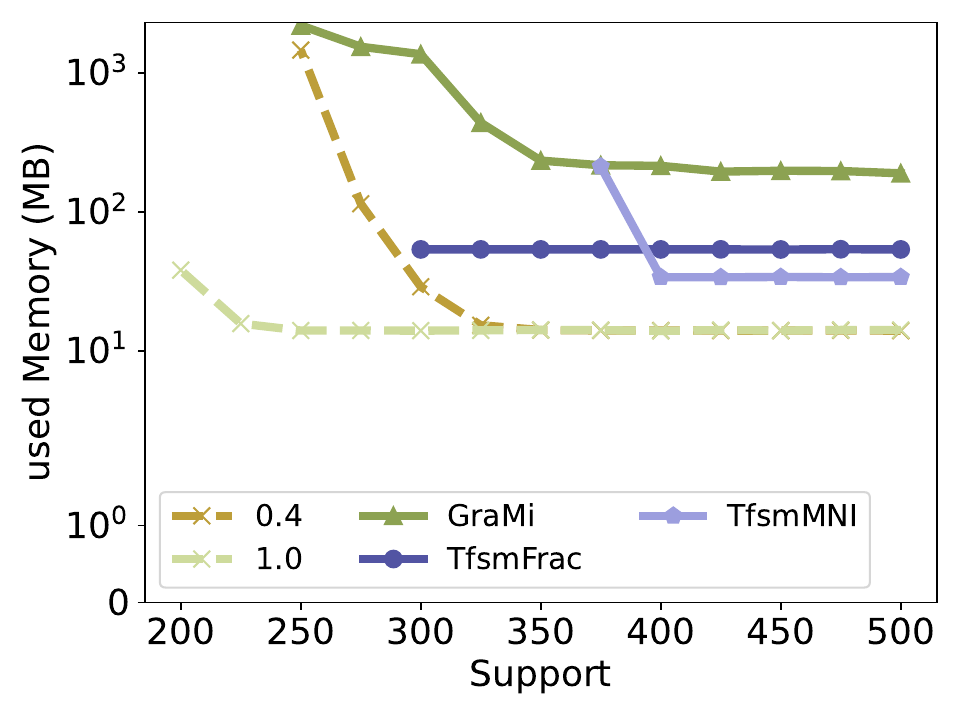}
    \caption{wiki-Vote}
  \end{subfigure}
  \begin{subfigure}{0.245\linewidth}
    \includegraphics[width=1.0\linewidth]{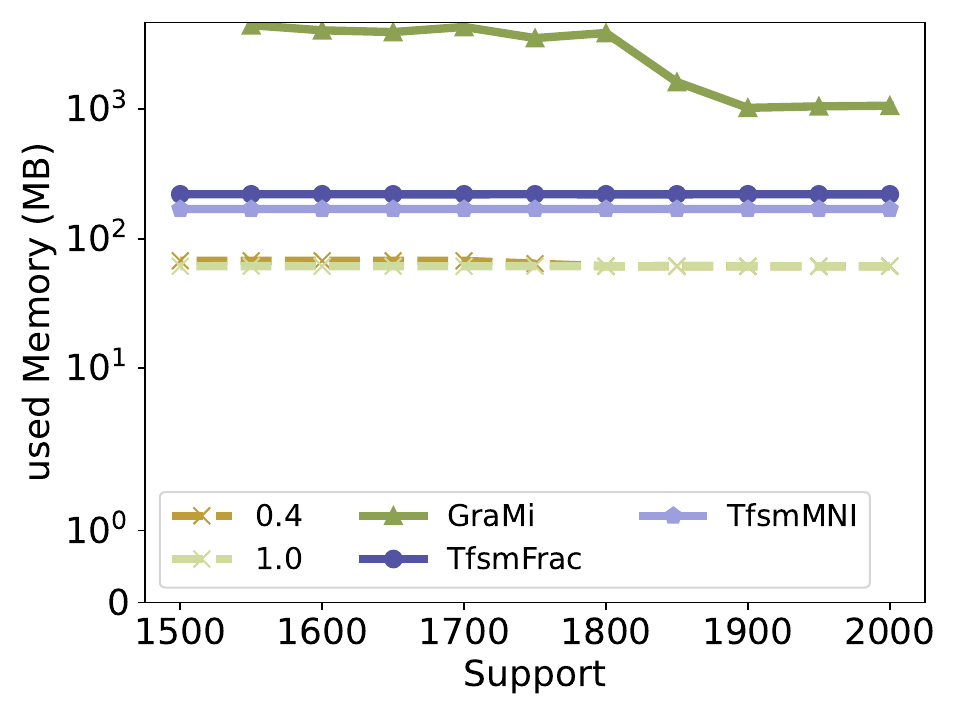}
    \caption{Epinions}
  \end{subfigure}

  \caption{Memory Usage}
  \label{fig:log:MemoryUsage}
  \Description{}
\end{figure*}

\begin{figure*}[htb]
  \centering
  \begin{subfigure}{0.245\linewidth}
    \includegraphics[width=1.0\linewidth]{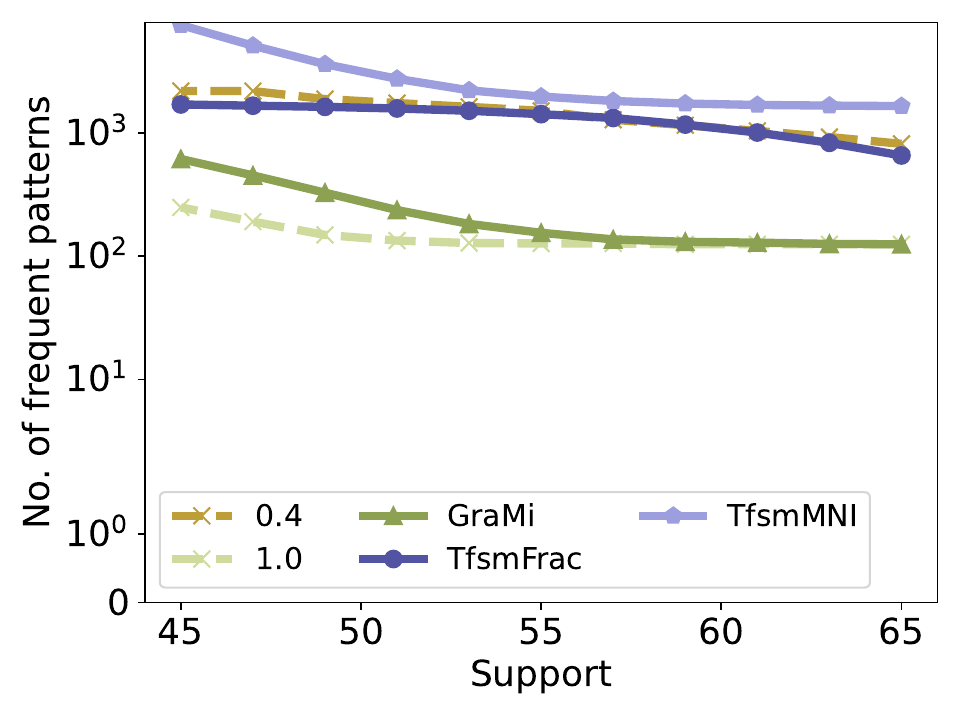}
    \caption{Gnutella}
    \label{fig:log:freq:Gnu}
  \end{subfigure}
  \begin{subfigure}{0.245\linewidth}
    \includegraphics[width=1.0\linewidth]{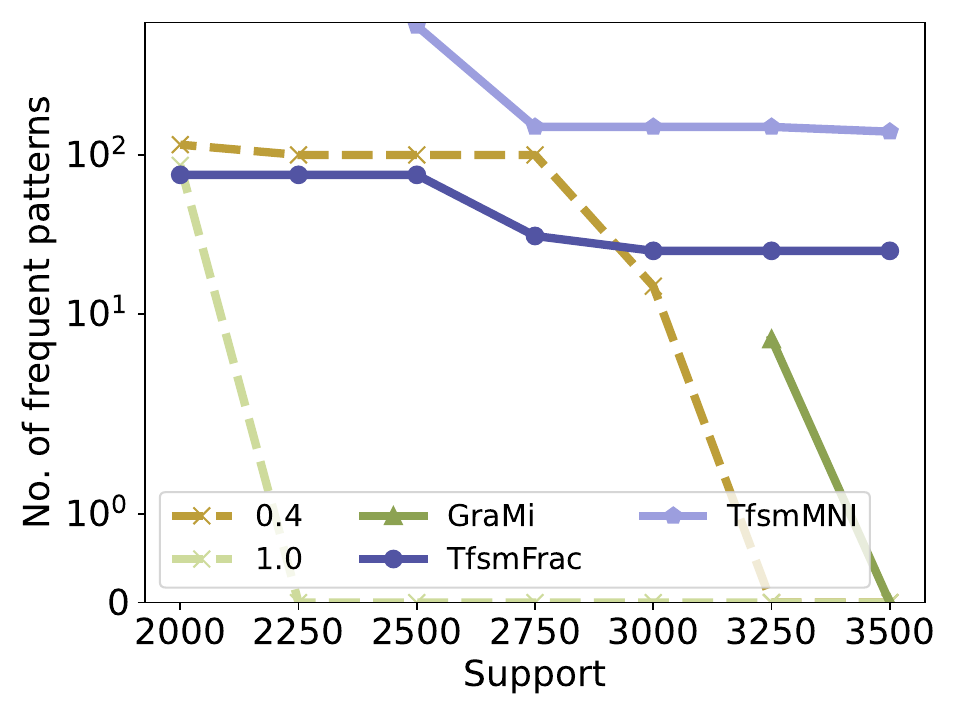}
    \caption{Slashdot}
    \label{fig:log:freq:dot}
  \end{subfigure}
  \begin{subfigure}{0.245\linewidth}
    \includegraphics[width=1.0\linewidth]{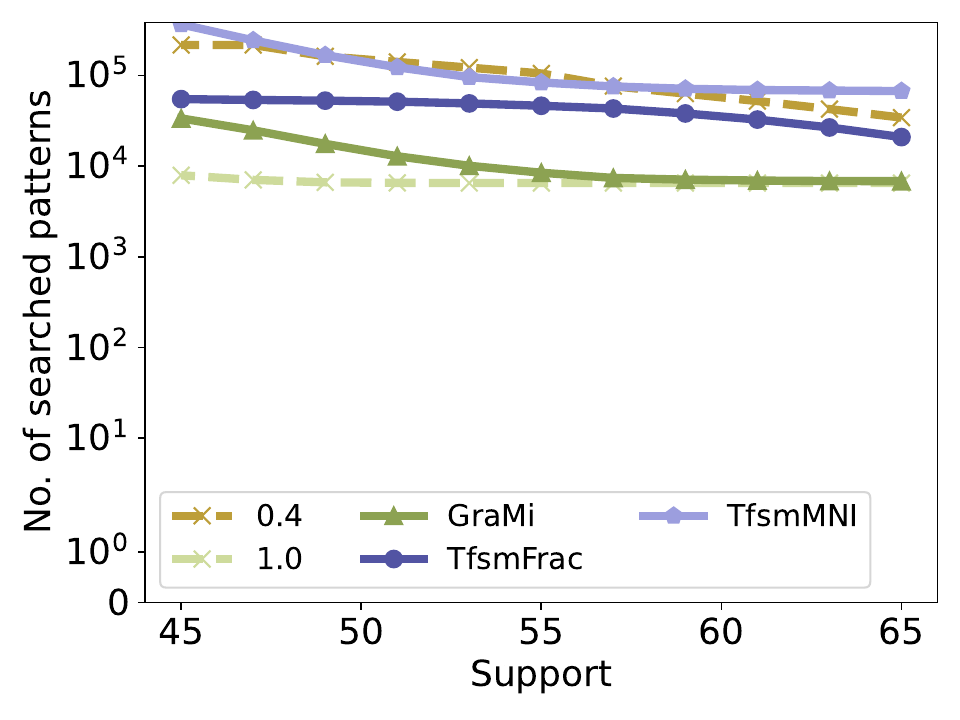}
    \caption{Gnutella}
    \label{fig:log:srch:Gnu}
  \end{subfigure}
  \begin{subfigure}{0.245\linewidth}
    \includegraphics[width=1.0\linewidth]{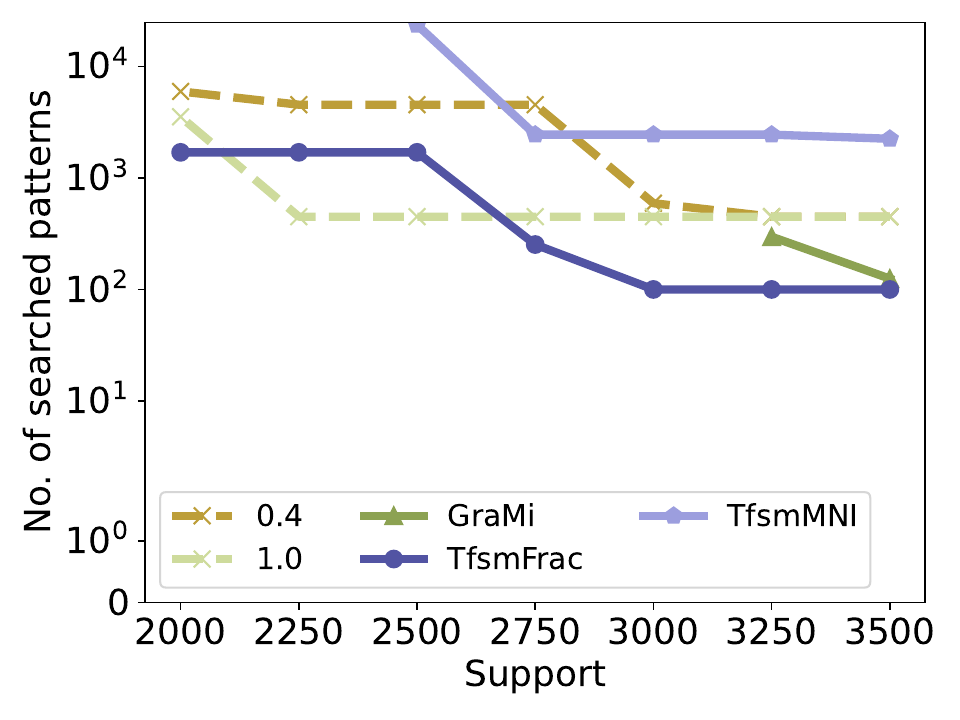}
    \caption{Slashdot}
    \label{fig:log:srch:dot}
  \end{subfigure}
  \caption{Frequent (a, b) and Searched (c, d) Patterns}
  \label{fig:log:FrequentSearched}
  \Description{}
\end{figure*}


We evaluate the performance of \ac{flexis}, GraMi, and T-FSM in terms of the time taken to compute frequent subgraphs. The outcomes of both systems were recorded, with a predetermined time limit of $30$ minutes. For the sake of clarity for all the other subsequent graphs, we denote GraMi as \qoutes{GraMi},  T-FSM fractional-based algorithm is shown as \qoutes{TfsmFrac} while T-FSM \ac{mni} is shown as \qoutes{TfsmMNI}, while numeric identifiers are used for different configurations of \ac{flexis}, unless mentioned otherwise. Since, \ac{flexis} is characterized by a parameter range spanning $\left[0.0, 1.0\right]$ (\ac{lerp}), we use the value chosen for the representation.

\subsection{Methodology to select \ac{lerp}}
It is commonly known that the number of times a pattern occurs is overestimated by \ac{mni} and as a consequence, the number of patterns classified as frequent is more than the actual value for GraMi. And a similar explanation also holds for T-FSM-MNI. If we take a closer look at the functioning of T-FSM using fractional-score, there is still a considerable overestimation. However, this overestimation is much less than that of GraMi. 
In contrast, our method spans from over-estimating to under-estimating, with over-estimation happening at a lower \ac{lerp} value and under-estimation possibly at a higher value. Thus, we postulate that we will merge closer to GraMi at a lower value of \ac{lerp}, to the T-FSM-MNI at a similar value and to T-FSM, at the same or higher value, depending on the efficiency of fractional-score with those datasets. We follow the following method to determine the most suitable \ac{lerp} value to compare with GraMi and T-FSM.

\begin{itemize}
    \item \textbf{Comparison with different \ac{lerp} values:} We show the time taken (\cref{fig:slider:time}) and the number of frequent patterns (\cref{fig:slider:frequent}) for varied values of $\ac{lerp}$. As expected, the number of patterns found decreases with an increase in $\ac{lerp}$, because of a reduction in overestimation. Also, the time required to find the patterns decreases with an increase in $\ac{lerp}$, as the number of patterns searched for decreases. Thus, to be fair in comparison with GraMi and T-FSM, we need to determine the value of \ac{lerp} in which the number of patterns found matches with that of GraMi and T-FSM.
    \item \textbf{Choosing a suitable \ac{lerp} value: }   From our experiments, as shown in the \cref{fig:log:TimeTaken:Gnutella} most of the values of the execution time of T-FSM (lowest among the other comparing methods) falls closer to our slider value of 0.5. But, as can be seen in the \cref{tab:similarity} the number of frequent patterns in all the three algorithms, namely \ac{flexis}, GraMi and T-FSM merge at 0.4. To show the similarity in the frequent patterns determined, we used graph isomorphim, and present the results in \cref{tab:similarity}. In the table, $f_g, f_f, f_t$ refers to frequent patterns generated by GraMi, \ac{flexis}, and T-FSM respectively, and $V_p$ represents the number of vertices. The isomorophism between, GraMi and \ac{flexis}, T-FSM and \ac{flexis} are shown as, $f_f \cap f_g$ and $f_f \cap f_t$ respectively.  We show only one dataset because of the lack of space. Thus, for the fairness of comparison for the rest of the paper, we will consider the \ac{lerp} value as $0.4$.
\end{itemize}

\begin{figure}[htb]
    \centering
    \begin{subfigure}{0.45\linewidth}
        \includegraphics[width=1.0\linewidth]{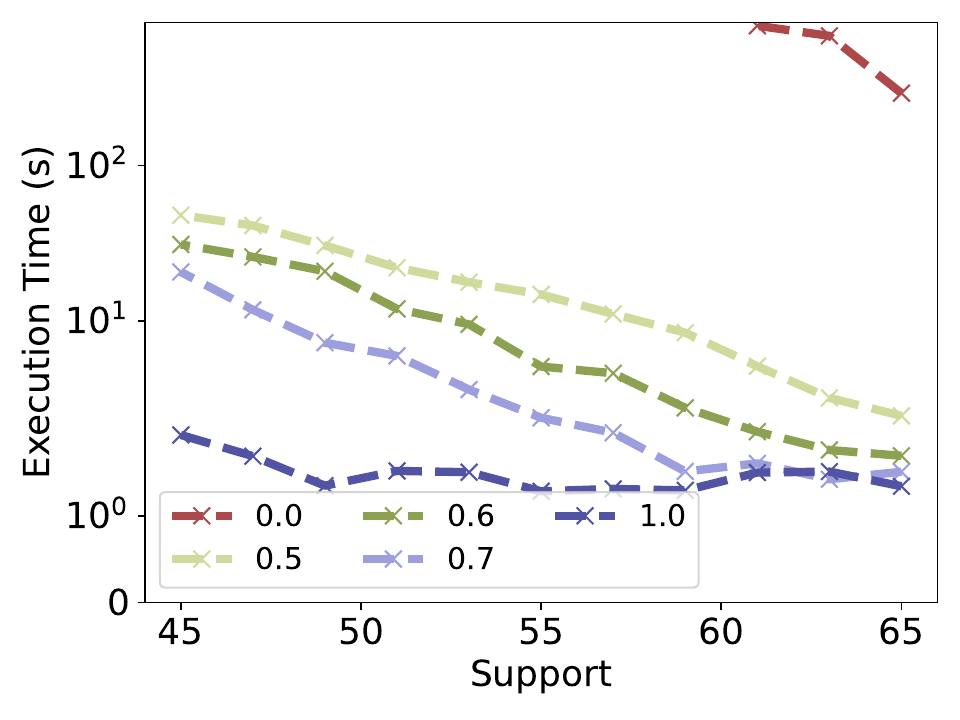}
        \caption{Time taken}
        \label{fig:slider:time}
    \end{subfigure}
    \begin{subfigure}{0.45\linewidth}
        \includegraphics[width=1.0\linewidth]{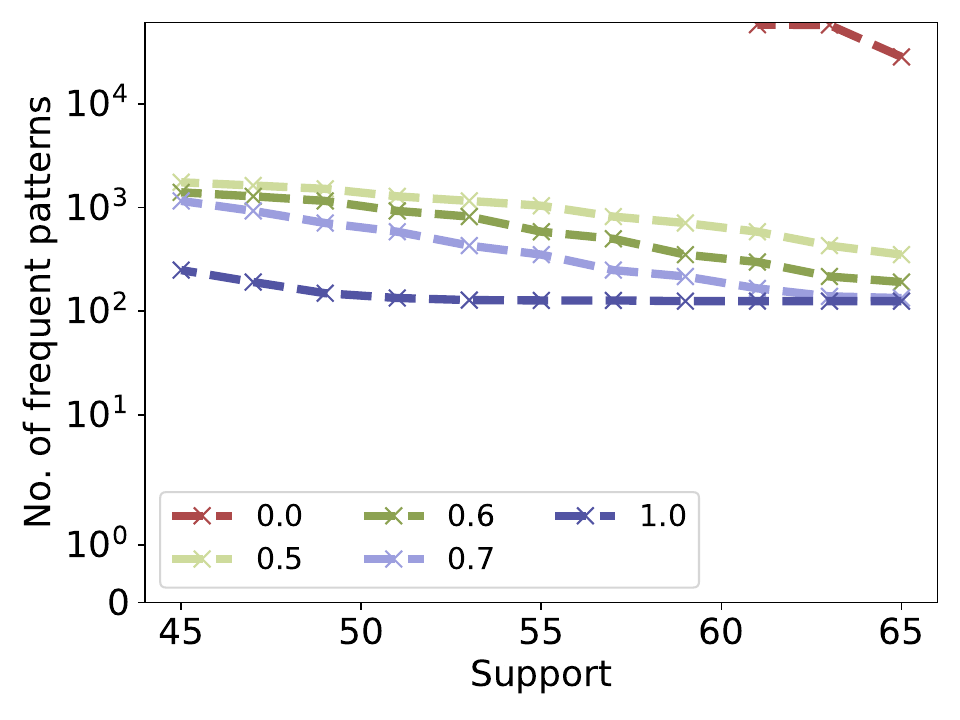}
        \caption{Frequent patterns}
        \label{fig:slider:frequent}
    \end{subfigure}
    \caption{Various slider values with Gnutella}
    \label{fig:slider:gnutella}
    \Description{}
\end{figure}


\subsection{Execution Time}
Execution Time for the Gnutella dataset, across varied slider values, along with GraMi and T-FSM is as shown in \cref{fig:log:TimeTaken:Gnutella}. For MiCo, Epinions, Slashdot, and wiki-Vote datagraphs in \cref{fig:log:TimeTaken}. 
Across Epinions, Slashdot and wiki-Vote, slider $0.4$ performs with an average speed up of 3.02x over T-FSM fractional and 10.58x over GraMi.
While with Gnutella, slider $0.4$ with an average speedup of 1.97x over T-FSM fractional and 7.90x over GraMi.
Notably, the number of frequent subgraphs exhibits a growth pattern as the support value is decreased.

\vspace{-10pt}
\subsection{Memory Usage} \label{sec:results:memoryUsage}
\Cref{fig:log:MemoryUsage} presents memory utilization patterns observed during the processing of the Epinions and wiki-Vote data graphs by \ac{flexis}, GraMi and T-FSM. It is important to note that instances, where execution was prematurely terminated due to reaching predefined timeout constraints, do not have corresponding memory utilization values represented. As shown in \cref{fig:log:MemoryUsage}, GraMi and T-FSM Fractional uses 8.79x and 1.01x memory than \ac{flexis} ($\lambda = 0.4$) in average across wiki-Vote, Epinions, and Slashdot. We postulate that the differences in memory utilization patterns emerge between the three systems due to their distinct storage strategies. Specifically, GraMi maintains an expansive repository of both frequent and infrequent graphs, hence higher memory demand. T-FSM is built on top of GraMi and has similar memory requirements, but since T-FSM performs a memory-bound search in the data graph it has lower memory requirements when compared to GraMi. On the contrary, \ac{flexis} approach stores only the frequent graphs at each level. Since, for the determination of the potential candidates for $k$-size pattern only requires $k-1$-size pattern, thereby adapting a more resource-efficient memory management framework. It is important to note that with the addition of Gnutella, GraMi and T-FSM Fractional uses 5.83x and 0.64x more memory than \ac{flexis}. The decrease in memory improvement achieved by \ac{flexis} can be attributed to more searches required for the Gnutella dataset. This leads to a higher number of frequent and searched patterns.

We can also observe, that as the support value decreases, the memory utilization proportionately increases. This phenomenon occurs due to the greater number of patterns being recognized as frequent as the support threshold decreases. This trend is consistent across both GraMi and our method \ac{flexis}. In the context of our approach, a reduction in the \ac{lerp} value corresponds to an elevated count of frequent patterns, consequently translating to an increase in memory requirement. It is crucial to acknowledge that this increased memory allocation, with a decrease in support, in our approach is essential, as the nature of the mining problem necessitates the retention of identified frequent subgraphs. In stark contrast to GraMi, our approach refrains from incurring any excess memory overhead beyond that required to store the indispensable patterns. Thus, our approach effectively harnesses considerably less memory in comparison to GraMi. In comparison to T-FSM, \ac{flexis} sometimes have higher memory consumption (Gnutella and MiCo), while other times ours are more efficient (Epinions, Slashdot, and wiki-Vote). This is because the memory utilization reported is the maximum utilization and it is possible that at a particular time in the entire process, we might have had a slightly higher memory utilization.

\begin{table}
  \centering
  \caption{Searched patterns}
  \label{tab:searchedPatterns}
  No. of searched patterns for different support values for Gnutella.
  $|S_g|$ for MNI,
  $|S_f|$ for \acs{mal}, and 
  $|S_t|$ for MNI-Fractional.
  T-FSM's version of MNI is used due to its higher efficiency than GraMi~\cite{yuan:2023:tfsm}\\
  \begin{tabular}{cccc}
    \toprule
      & $|S_g|$ & $|S_f|$ & $|S_t|$ \\
    support &  &  &  \\
    \midrule
     57 & 75043 & 34236 & 43110 \\
     59 & 70910 & 27022 & 38035 \\
     61 & 68727 & 19865 & 32573 \\
     63 & 67853 & 13044 & 26657 \\
     65 & 67025 & 10431 & 20879 \\
    \bottomrule
\end{tabular}
    
\end{table}

\subsection{Frequent Subgraph Instances}
\cref{tab:searchedPatterns} shows that we conducted fewer searches when compared to GraMi and T-FSM algorithms, for a $\lambda$ value of $0.5$. The reason we conducted fewer searches is because of our way of choosing which patterns to look for.
We focused only on exploring graphs that could be made from the frequent graphs we found earlier. In contrast to GraMi, which spends time looking at all the possible extensions that could come from the frequent graphs in previous steps.

However, \ac{flexis} searches more candidate patterns as the effective threshold $\tau$ reduces as $\lambda$ decreases. Specifically, \ac{flexis} ($\lambda=0.4$) searches 7.2x and 5.7x more than GraMi and T-FSM fractional. This might be because, we find more frequent patterns in at $\lambda=0.4$ when compared to GraMi and T-FSM, as shown in \cref{fig:log:freq:Gnu,fig:log:freq:dot}. Precisely we found 8.8x and 2.4x more than frequent patterns than GraMi and T-FSM. Note as the number of found patterns increases, the number of patterns that can be combined to form a new candidate pattern also increases.

\begin{table}[htb]
    \centering
    \caption{Similarity with \ac{flexis} $\lambda = 0.4$ on Epinions}
    \label{tab:similarity}
    \begin{tabular}{|c|c|c|c|c|c|c|}
  \toprule
  $\sigma$ & $\abs{V_p}$ & $|f_g|$ & $|f_f|$ & $|f_t|$ & $|f_f\cap f_g|$ & $|f_f\cap f_t|$ \\
  \midrule
  1600 & 2 & 125 & 100 & 75 & 100 & 50 \\ \hline
  1700 & 2 & 108 & 100 & 75 & 86 & 50 \\

  \bottomrule
\end{tabular}

\end{table}

\subsection{Difference in identified Frequent Patterns}

As can be noticed in~\cref{tab:similarity}, there is not a complete intersection between both, this can be attributed to the following reasons,
\begin{itemize}
    \item \textbf{Difference in the approach:} Our method uses an undirected data processing and directed matcher method, while GraMi uses fully directed search. So there is a possibility that we are getting some frequent patterns, which GraMi does not find.
    \item \textbf{Difference in over-estimation:} In lower slider values, \ac{mal} tends to over-estimate due to the approximation applied to the threshold value. In contrast, the overestimation of GraMi is due to the recounting of the same vertices, considering automorphism and so on. Thus, the over-estimated patterns might not match between \ac{flexis} and GraMi.
\end{itemize}


\newpage
\section{Related Work}~\label{sec:related}

The solution to the FSM problem can be largely derived from a two-step process, primarily extended from the work of ~\cite{nguyen:2022:subgraph}. The initial step, referred to as the \sqoutes{Generation step}, involves the creation of all possible candidate patterns and potential subgraph candidates. Following this, the candidate subgraph space is streamlined by pruned automorphisms. Lastly, the \sqoutes{Metric step} is the number of isomorphic instances among the selected candidates within the data graph, determining whether it exceeds predefined support. Many of the algorithms presented in the context of Subgraph Isomorphism focus on modifying one or more of these components to enhance overall performance.

In the generation step, all possible subgraphs can be achieved through two main methods: edge extension and vertex extension. In the edge extension methods, a candidate subgraph is derived from an existing frequent candidate by adding an edge. This process can be accomplished in two distinct ways. Initially, the path extension approach is employed, where vertices within the frequent subgraph are expanded by adding edges, as illustrated in previous works such as Peregrine, GSpan, Pangolin, and Arabesque~\cite{jamshidi:2020:peregrine,yan:2002:gspan,chen:2020:pangolin, teixeira:2015:arabesque}. Alternatively, the utilization of level-wise methods becomes prominent, wherein two frequent subgraphs of order $k$ combine to form a subgraph of order $k+1$~\cite{kuramochi:2005:finding}. In the context of vertex extension methods, the frequent subgraphs are expanded by adding a vertex~\cite{chen:2020:pangolin, teixeira:2015:arabesque}. Notably, there has been limited exploration into the extension of level-wise notions within the context of vertex extension scenarios. In this paper, we introduce a novel approach that integrates the level-wise methodology into vertex extension methods.  

Within the generation step, the redundant and automorphic candidate graphs must be eliminated. The common practice involves the utilization of canonical representations to eliminate redundant candidate graphs that are generated. Diverse canonical representations exist, including the Minimum DFS Code representation~\cite{yan:2002:gspan}, Bliss~\cite{junttila:2007:engineering}, Eigen Value representations~\cite{zhao:2020:kaleido}, Canonical Adjacency matrix representations, as well as Breadth First Canonical String and Depth First Canonical String representations~\cite{jiang:2013:survey}, among others. Given its exceptional efficiency, Bliss is selected as the method of choice in this paper.

The last step of FSM is the `Metric' phase, wherein the generated candidates are compared with the data graphs to ascertain the frequency of occurrence (matches) of a specific candidate within the data graph. This comparison also involves assessing whether the match count exceeds a predefined threshold. Subgraph isomorphism methods are employed for match determination. Varied iterations of this phase exist within the literature. The gold standard is utilizing the Maximum Independent Set~\cite{fiedler:2007:support}, which identifies the maximal number of disjoint subgraph isomorphisms in the primary graph, an NP-complete problem in itself. An alternative, SATMargin~\cite{liu:2022:satmargin}, employs random walks, framing the subgraph mining challenge as an SAT problem and employing an SAT algorithm (such as CryptoMiniSat~\cite{soos:2019:cryptominisat}) for its resolution. Notably, the most prevailing approximation method is \ac{mni}~\cite{bringmann:2008:frequent}, which tends to overestimate the number of isomorphic patterns by allowing certain overlaps. \ac{mni} finds extensive utilization either directly or with subsequent optimization strategies~\cite{elseidy:2014:grami,jamshidi:2020:peregrine,chen:2020:pangolin,chen:2021:sandslash}. In the context of GraMi~\cite{elseidy:2014:grami}, graph mining is conceptualized as a Constraint Satisfaction Problem (CSP) that uses \ac{mni} to prune the search space. Other metrics that are an extension to \ac{mni}, such as a fractional score, which alleviates the over-estimation of \ac{mni} have also been proposed~\cite{yuan:2023:tfsm}. While this method succeeds in reducing the overestimation to a limit, it still shows a substantial overestimation. Therefore, this paper introduces an innovative metric termed 'Maximal Independent Set,' an approximation of the Maximum Independent Set that provides the users with a slider to control the overlap that is required with the Maximum Independent set. While we might also overestimate, it can be brought to the minimum if the user chooses to. Notably, our proposed method diverges from the concepts of maximal frequent subgraph as discussed in SATMargin~\cite{liu:2022:satmargin} and Margin~\cite{thomas:2010:margin}, wherein the focus is on determining the maximum feasible size of a frequent subgraph, as opposed to utilizing it as a metric, as shown in our study.

While tangential to the core focus of this paper, the subsequent passages are included to provide a definitive categorization of our study.

FSM algorithms can be classified into two distinct categories: 1) Transactional methods, and 2) Single-graph methods. Transactional algorithms~\cite{liu:2022:satmargin, thomas:2010:margin}, define a frequent pattern as one that occurs across a set number of graphs. In other words, within a collection of $N$ data graphs, a pattern is deemed frequent if it emerges in more than \ac{support} data graphs. Conversely, the single graph paradigm considers a pattern frequent if it repeats a predefined number of times within a single large graph~\cite{elseidy:2014:grami, nguyen:2021:method, mawhirter:2021:graphzero, jamshidi:2020:peregrine, chen:2022:efficient}. This paper squarely falls within the realm of Single-graph methods. Notably, there exist other methodological extensions within FSM that aim to ascertain precise support counts~\cite{nguyen:2020:fast}, solely consider closed-form solutions~\cite{nguyen:2021:method} or delve into weighted subgraph mining~\cite{le:2020:mining}. These extensions, however, lie outside the scope of our current investigation.

Moreover, Subgraph Mining can be achieved through diverse ways, spanning CPU-based approaches~\cite{elseidy:2014:grami, mawhirter:2019:automine, mawhirter:2021:graphzero, aberger:2017:emptyheaded, shi:2020:graphpi, jamshidi:2020:peregrine}, GPU-accelerated techniques~\cite{chen:2020:pangolin, guo:2020:gpu, guo:2020:exploiting, chen:2022:efficient}, Distributed Systems~\cite{salem:2021:rasma}, near-memory architectures~\cite{dai:2022:dimmining}, and out-of-core systems~\cite{zhao:2020:kaleido}. In this present study, we focus exclusively on CPU-based architectures. However, it is important to note that our proposed method exhibits adaptability for potential extension to GPU and other architectural frameworks.

\section{Conclusion}~\label{sec:conclusion}
In this paper, we introduce a novel metric based on maximal independent sets (MIS) that enables users to fine-tune the level of overlap they require with the original results. This level of customizability can be invaluable, particularly in applications where precision is paramount. We also proposed an innovative technique for efficiently generating the candidate search space by merging previously identified frequent subgraphs. This approach significantly reduced the computational overhead associated with graph mining, making it more accessible and practical. By incorporating these methods we showed that our method performs better than the existing state-of-art GraMi and T-FSM. In essence, our contributions pave the way for more effective and versatile graph mining techniques, expanding their utility across a wide spectrum of applications.

\bibliographystyle{bibstyle/ACM-Reference-Format}
\bibliography{references}


\begin{thebibliography}{47}


\ifx \showCODEN    \undefined \def \showCODEN     #1{\unskip}     \fi
\ifx \showDOI      \undefined \def \showDOI       #1{#1}\fi
\ifx \showISBNx    \undefined \def \showISBNx     #1{\unskip}     \fi
\ifx \showISBNxiii \undefined \def \showISBNxiii  #1{\unskip}     \fi
\ifx \showISSN     \undefined \def \showISSN      #1{\unskip}     \fi
\ifx \showLCCN     \undefined \def \showLCCN      #1{\unskip}     \fi
\ifx \shownote     \undefined \def \shownote      #1{#1}          \fi
\ifx \showarticletitle \undefined \def \showarticletitle #1{#1}   \fi
\ifx \showURL      \undefined \def \showURL       {\relax}        \fi
\providecommand\bibfield[2]{#2}
\providecommand\bibinfo[2]{#2}
\providecommand\natexlab[1]{#1}
\providecommand\showeprint[2][]{arXiv:#2}

\bibitem[\protect\citeauthoryear{Aberger, Lamb, Tu, N{\"o}tzli, Olukotun, and R{\'e}}{Aberger et~al\mbox{.}}{2017}]%
        {aberger:2017:emptyheaded}
\bibfield{author}{\bibinfo{person}{Christopher~R Aberger}, \bibinfo{person}{Andrew Lamb}, \bibinfo{person}{Susan Tu}, \bibinfo{person}{Andres N{\"o}tzli}, \bibinfo{person}{Kunle Olukotun}, {and} \bibinfo{person}{Christopher R{\'e}}.} \bibinfo{year}{2017}\natexlab{}.
\newblock \showarticletitle{Emptyheaded: A relational engine for graph processing}.
\newblock \bibinfo{journal}{\emph{ACM Transactions on Database Systems (TODS)}} \bibinfo{volume}{42}, \bibinfo{number}{4} (\bibinfo{year}{2017}), \bibinfo{pages}{1--44}.
\newblock


\bibitem[\protect\citeauthoryear{Bindu and Thilagam}{Bindu and Thilagam}{2016}]%
        {bindu:2016:mining}
\bibfield{author}{\bibinfo{person}{PV Bindu} {and} \bibinfo{person}{P~Santhi Thilagam}.} \bibinfo{year}{2016}\natexlab{}.
\newblock \showarticletitle{Mining social networks for anomalies: Methods and challenges}.
\newblock \bibinfo{journal}{\emph{Journal of Network and Computer Applications}}  \bibinfo{volume}{68} (\bibinfo{year}{2016}), \bibinfo{pages}{213--229}.
\newblock


\bibitem[\protect\citeauthoryear{Bringmann and Nijssen}{Bringmann and Nijssen}{2008}]%
        {bringmann:2008:frequent}
\bibfield{author}{\bibinfo{person}{Bj{\"o}rn Bringmann} {and} \bibinfo{person}{Siegfried Nijssen}.} \bibinfo{year}{2008}\natexlab{}.
\newblock \showarticletitle{What is frequent in a single graph?}. In \bibinfo{booktitle}{\emph{Advances in Knowledge Discovery and Data Mining: 12th Pacific-Asia Conference, PAKDD 2008 Osaka, Japan, May 20-23, 2008 Proceedings 12}}. Springer, \bibinfo{publisher}{Springer Berlin Heidelberg}, \bibinfo{address}{Berlin, Heidelberg}, \bibinfo{pages}{858--863}.
\newblock
\urldef\tempurl%
\url{https://doi.org/10.1007/978-3-540-68125-0_84}
\showDOI{\tempurl}


\bibitem[\protect\citeauthoryear{Carletti, Foggia, Greco, Saggese, and Vento}{Carletti et~al\mbox{.}}{2018}]%
        {carletti:2018:vf3}
\bibfield{author}{\bibinfo{person}{Vincenzo Carletti}, \bibinfo{person}{Pasquale Foggia}, \bibinfo{person}{Antonio Greco}, \bibinfo{person}{Alessia Saggese}, {and} \bibinfo{person}{Mario Vento}.} \bibinfo{year}{2018}\natexlab{}.
\newblock \showarticletitle{The VF3-light subgraph isomorphism algorithm: when doing less is more effective}. In \bibinfo{booktitle}{\emph{Structural, Syntactic, and Statistical Pattern Recognition: Joint IAPR International Workshop, S+ SSPR 2018, Beijing, China, August 17--19, 2018, Proceedings 9}}. Springer, \bibinfo{publisher}{Springer Berlin Heidelberg}, \bibinfo{address}{Berlin, Heidelberg}, \bibinfo{pages}{315--325}.
\newblock
\urldef\tempurl%
\url{https://doi.org/10.1007/978-3-319-97785-0_30}
\showDOI{\tempurl}


\bibitem[\protect\citeauthoryear{Carletti, Foggia, Greco, Vento, and Vigilante}{Carletti et~al\mbox{.}}{2019}]%
        {carletti:2019:vf3}
\bibfield{author}{\bibinfo{person}{Vincenzo Carletti}, \bibinfo{person}{Pasquale Foggia}, \bibinfo{person}{Antonio Greco}, \bibinfo{person}{Mario Vento}, {and} \bibinfo{person}{Vincenzo Vigilante}.} \bibinfo{year}{2019}\natexlab{}.
\newblock \showarticletitle{VF3-Light: a lightweight subgraph isomorphism algorithm and its experimental evaluation}.
\newblock \bibinfo{journal}{\emph{Pattern Recognition Letters}}  \bibinfo{volume}{125} (\bibinfo{year}{2019}), \bibinfo{pages}{591--596}.
\newblock
\urldef\tempurl%
\url{https://doi.org/10.1016/j.patrec.2019.07.001}
\showDOI{\tempurl}


\bibitem[\protect\citeauthoryear{Chen et~al\mbox{.}}{Chen et~al\mbox{.}}{2022}]%
        {chen:2022:efficient}
\bibfield{author}{\bibinfo{person}{Xuhao Chen} {et~al\mbox{.}}} \bibinfo{year}{2022}\natexlab{}.
\newblock \showarticletitle{Efficient and Scalable Graph Pattern Mining on $\{$GPUs$\}$}. In \bibinfo{booktitle}{\emph{16th USENIX Symposium on Operating Systems Design and Implementation (OSDI 22)}}. \bibinfo{publisher}{UNENIX Association}, \bibinfo{address}{2560 Ninth St. Suite 215 Berkeley, CA}, \bibinfo{pages}{857--877}.
\newblock


\bibitem[\protect\citeauthoryear{Chen, Dathathri, Gill, Hoang, and Pingali}{Chen et~al\mbox{.}}{2021}]%
        {chen:2021:sandslash}
\bibfield{author}{\bibinfo{person}{Xuhao Chen}, \bibinfo{person}{Roshan Dathathri}, \bibinfo{person}{Gurbinder Gill}, \bibinfo{person}{Loc Hoang}, {and} \bibinfo{person}{Keshav Pingali}.} \bibinfo{year}{2021}\natexlab{}.
\newblock \showarticletitle{Sandslash: a two-level framework for efficient graph pattern mining}. In \bibinfo{booktitle}{\emph{Proceedings of the ACM International Conference on Supercomputing}}. ACM, \bibinfo{publisher}{Association for Computing Machinery}, \bibinfo{address}{New York, NY, USA}, \bibinfo{pages}{378--391}.
\newblock
\urldef\tempurl%
\url{https://doi.org/10.48550/arXiv.2011.03135}
\showDOI{\tempurl}


\bibitem[\protect\citeauthoryear{Chen, Dathathri, Gill, and Pingali}{Chen et~al\mbox{.}}{2020}]%
        {chen:2020:pangolin}
\bibfield{author}{\bibinfo{person}{Xuhao Chen}, \bibinfo{person}{Roshan Dathathri}, \bibinfo{person}{Gurbinder Gill}, {and} \bibinfo{person}{Keshav Pingali}.} \bibinfo{year}{2020}\natexlab{}.
\newblock \showarticletitle{Pangolin: An efficient and flexible graph mining system on cpu and gpu}.
\newblock \bibinfo{journal}{\emph{Proceedings of the VLDB Endowment}} \bibinfo{volume}{13}, \bibinfo{number}{8} (\bibinfo{year}{2020}), \bibinfo{pages}{1190--1205}.
\newblock


\bibitem[\protect\citeauthoryear{Clark and Holton}{Clark and Holton}{2005}]%
        {Clark:2005:AFirstLookAtGraphTheory}
\bibfield{author}{\bibinfo{person}{John Clark} {and} \bibinfo{person}{Derek~Allan Holton}.} \bibinfo{year}{2005}\natexlab{}.
\newblock \bibinfo{booktitle}{\emph{A first look at graph theory}}.
\newblock \bibinfo{publisher}{World Scientific}, \bibinfo{address}{Singapore}.
\newblock


\bibitem[\protect\citeauthoryear{Dai, Zhu, Fu, Wei, Wang, Li, Xie, Yang, and Wang}{Dai et~al\mbox{.}}{2022}]%
        {dai:2022:dimmining}
\bibfield{author}{\bibinfo{person}{Guohao Dai}, \bibinfo{person}{Zhenhua Zhu}, \bibinfo{person}{Tianyu Fu}, \bibinfo{person}{Chiyue Wei}, \bibinfo{person}{Bangyan Wang}, \bibinfo{person}{Xiangyu Li}, \bibinfo{person}{Yuan Xie}, \bibinfo{person}{Huazhong Yang}, {and} \bibinfo{person}{Yu Wang}.} \bibinfo{year}{2022}\natexlab{}.
\newblock \showarticletitle{Dimmining: pruning-efficient and parallel graph mining on near-memory-computing}. In \bibinfo{booktitle}{\emph{Proceedings of the 49th Annual International Symposium on Computer Architecture}}. \bibinfo{publisher}{Association for Computing Machinery}, \bibinfo{address}{New York, NY, USA}, \bibinfo{pages}{130--145}.
\newblock
\urldef\tempurl%
\url{https://doi.org/10.1145/3470496.3527388}
\showDOI{\tempurl}


\bibitem[\protect\citeauthoryear{Elseidy, Abdelhamid, Skiadopoulos, and Kalnis}{Elseidy et~al\mbox{.}}{2014}]%
        {elseidy:2014:grami}
\bibfield{author}{\bibinfo{person}{Mohammed Elseidy}, \bibinfo{person}{Ehab Abdelhamid}, \bibinfo{person}{Spiros Skiadopoulos}, {and} \bibinfo{person}{Panos Kalnis}.} \bibinfo{year}{2014}\natexlab{}.
\newblock \showarticletitle{Grami: Frequent subgraph and pattern mining in a single large graph}.
\newblock \bibinfo{journal}{\emph{Proceedings of the VLDB Endowment}} \bibinfo{volume}{7}, \bibinfo{number}{7} (\bibinfo{year}{2014}), \bibinfo{pages}{517--528}.
\newblock
\urldef\tempurl%
\url{https://doi.org/10.14778/2732286.2732289}
\showDOI{\tempurl}


\bibitem[\protect\citeauthoryear{Fiedler and Borgelt}{Fiedler and Borgelt}{2007}]%
        {fiedler:2007:support}
\bibfield{author}{\bibinfo{person}{Mathias Fiedler} {and} \bibinfo{person}{Christian Borgelt}.} \bibinfo{year}{2007}\natexlab{}.
\newblock \showarticletitle{Support computation for mining frequent subgraphs in a single graph}. In \bibinfo{booktitle}{\emph{MLG}}. \bibinfo{publisher}{Association for Computing Machinery}, \bibinfo{address}{New York, NY, USA}.
\newblock


\bibitem[\protect\citeauthoryear{Guo, Li, Sha, He, Xiao, and Tan}{Guo et~al\mbox{.}}{2020b}]%
        {guo:2020:gpu}
\bibfield{author}{\bibinfo{person}{Wentian Guo}, \bibinfo{person}{Yuchen Li}, \bibinfo{person}{Mo Sha}, \bibinfo{person}{Bingsheng He}, \bibinfo{person}{Xiaokui Xiao}, {and} \bibinfo{person}{Kian-Lee Tan}.} \bibinfo{year}{2020}\natexlab{b}.
\newblock \showarticletitle{Gpu-accelerated subgraph enumeration on partitioned graphs}. In \bibinfo{booktitle}{\emph{Proceedings of the 2020 ACM SIGMOD International Conference on Management of Data}}. \bibinfo{publisher}{Association for Computing Machinery}, \bibinfo{address}{New York, NY, USA}, \bibinfo{pages}{1067--1082}.
\newblock


\bibitem[\protect\citeauthoryear{Guo, Li, and Tan}{Guo et~al\mbox{.}}{2020a}]%
        {guo:2020:exploiting}
\bibfield{author}{\bibinfo{person}{Wentian Guo}, \bibinfo{person}{Yuchen Li}, {and} \bibinfo{person}{Kian-Lee Tan}.} \bibinfo{year}{2020}\natexlab{a}.
\newblock \showarticletitle{Exploiting reuse for GPU subgraph enumeration}.
\newblock \bibinfo{journal}{\emph{IEEE Transactions on Knowledge and Data Engineering}} \bibinfo{volume}{34}, \bibinfo{number}{9} (\bibinfo{year}{2020}), \bibinfo{pages}{4231--4244}.
\newblock


\bibitem[\protect\citeauthoryear{Jamshidi, Mahadasa, and Vora}{Jamshidi et~al\mbox{.}}{2020}]%
        {jamshidi:2020:peregrine}
\bibfield{author}{\bibinfo{person}{Kasra Jamshidi}, \bibinfo{person}{Rakesh Mahadasa}, {and} \bibinfo{person}{Keval Vora}.} \bibinfo{year}{2020}\natexlab{}.
\newblock \showarticletitle{Peregrine: a pattern-aware graph mining system}. In \bibinfo{booktitle}{\emph{Proceedings of the Fifteenth European Conference on Computer Systems}}. \bibinfo{publisher}{Association for Computing Machinery}, \bibinfo{address}{New York, NY, USA}, \bibinfo{pages}{1--16}.
\newblock
\urldef\tempurl%
\url{https://doi.org/10.48550/arXiv.2004.02369}
\showDOI{\tempurl}


\bibitem[\protect\citeauthoryear{Jiang, Coenen, and Zito}{Jiang et~al\mbox{.}}{2013}]%
        {jiang:2013:survey}
\bibfield{author}{\bibinfo{person}{Chuntao Jiang}, \bibinfo{person}{Frans Coenen}, {and} \bibinfo{person}{Michele Zito}.} \bibinfo{year}{2013}\natexlab{}.
\newblock \showarticletitle{A survey of frequent subgraph mining algorithms}.
\newblock \bibinfo{journal}{\emph{The Knowledge Engineering Review}} \bibinfo{volume}{28}, \bibinfo{number}{1} (\bibinfo{year}{2013}), \bibinfo{pages}{75--105}.
\newblock
\urldef\tempurl%
\url{https://doi.org/10.1017/S0269888912000331}
\showDOI{\tempurl}


\bibitem[\protect\citeauthoryear{Junttila and Kaski}{Junttila and Kaski}{2007a}]%
        {JunttilaKaski:ALENEX:2007:bliss}
\bibfield{author}{\bibinfo{person}{Tommi Junttila} {and} \bibinfo{person}{Petteri Kaski}.} \bibinfo{year}{2007}\natexlab{a}.
\newblock \showarticletitle{Engineering an efficient canonical labeling tool for large and sparse graphs}. In \bibinfo{booktitle}{\emph{Proceedings of the Ninth Workshop on Algorithm Engineering and Experiments and the Fourth Workshop on Analytic Algorithms and Combinatorics}}, \bibfield{editor}{\bibinfo{person}{David Applegate}, \bibinfo{person}{Gerth~St{\o}lting Brodal}, \bibinfo{person}{Daniel Panario}, {and} \bibinfo{person}{Robert Sedgewick}} (Eds.). \bibinfo{publisher}{SIAM}, \bibinfo{address}{Philadelphia}, \bibinfo{pages}{135--149}.
\newblock
\urldef\tempurl%
\url{https://doi.org/10.1137/1.9781611972870.13}
\showDOI{\tempurl}


\bibitem[\protect\citeauthoryear{Junttila and Kaski}{Junttila and Kaski}{2007b}]%
        {junttila:2007:engineering}
\bibfield{author}{\bibinfo{person}{Tommi Junttila} {and} \bibinfo{person}{Petteri Kaski}.} \bibinfo{year}{2007}\natexlab{b}.
\newblock \showarticletitle{Engineering an efficient canonical labeling tool for large and sparse graphs}. In \bibinfo{booktitle}{\emph{2007 Proceedings of the Ninth Workshop on Algorithm Engineering and Experiments (ALENEX)}}. SIAM, \bibinfo{publisher}{Society for Industrial and Applied Mathematics}, \bibinfo{address}{Philadelphia, Pennsylvania}, \bibinfo{pages}{135--149}.
\newblock
\urldef\tempurl%
\url{https://doi.org/10.5555/2791188.2791201}
\showDOI{\tempurl}


\bibitem[\protect\citeauthoryear{Junttila and Kaski}{Junttila and Kaski}{2011}]%
        {JunttilaKaski:TAPAS:2011:bliss}
\bibfield{author}{\bibinfo{person}{Tommi Junttila} {and} \bibinfo{person}{Petteri Kaski}.} \bibinfo{year}{2011}\natexlab{}.
\newblock \showarticletitle{Conflict Propagation and Component Recursion for Canonical Labeling}. In \bibinfo{booktitle}{\emph{Theory and Practice of Algorithms in (Computer) Systems -- First International {ICST} Conference, {TAPAS} 2011, Rome, Italy, April 18--20, 2011. Proceedings}} \emph{(\bibinfo{series}{Lecture Notes in Computer Science})}, \bibfield{editor}{\bibinfo{person}{Alberto Marchetti{-}Spaccamela} {and} \bibinfo{person}{Michael Segal}} (Eds.), Vol.~\bibinfo{volume}{6595}. \bibinfo{publisher}{Springer}, \bibinfo{address}{Berlin, Heidelberg}, \bibinfo{pages}{151--162}.
\newblock
\urldef\tempurl%
\url{https://doi.org/10.1007/978-3-642-19754-3\_16}
\showDOI{\tempurl}


\bibitem[\protect\citeauthoryear{Kingan}{Kingan}{2022}]%
        {kingan:2022:GraphsAndNetworks}
\bibfield{author}{\bibinfo{person}{Sandra~R. Kingan}.} \bibinfo{year}{2022}\natexlab{}.
\newblock \bibinfo{booktitle}{\emph{Graphs and Networks}}.
\newblock \bibinfo{publisher}{Wiley}, \bibinfo{address}{Hoboken, NJ}.
\newblock


\bibitem[\protect\citeauthoryear{Kong, Huang, Tan, and Liu}{Kong et~al\mbox{.}}{2022}]%
        {kong:2022:molecule}
\bibfield{author}{\bibinfo{person}{Xiangzhe Kong}, \bibinfo{person}{Wenbing Huang}, \bibinfo{person}{Zhixing Tan}, {and} \bibinfo{person}{Yang Liu}.} \bibinfo{year}{2022}\natexlab{}.
\newblock \showarticletitle{Molecule generation by principal subgraph mining and assembling}.
\newblock \bibinfo{journal}{\emph{Advances in Neural Information Processing Systems}}  \bibinfo{volume}{35} (\bibinfo{year}{2022}), \bibinfo{pages}{2550--2563}.
\newblock


\bibitem[\protect\citeauthoryear{Kuramochi and Karypis}{Kuramochi and Karypis}{2004}]%
        {kuramochi:2004:efficient}
\bibfield{author}{\bibinfo{person}{M. Kuramochi} {and} \bibinfo{person}{G. Karypis}.} \bibinfo{year}{2004}\natexlab{}.
\newblock \showarticletitle{An efficient algorithm for discovering frequent subgraphs}.
\newblock \bibinfo{journal}{\emph{IEEE Transactions on Knowledge and Data Engineering}} \bibinfo{volume}{16}, \bibinfo{number}{9} (\bibinfo{year}{2004}), \bibinfo{pages}{1038--1051}.
\newblock
\urldef\tempurl%
\url{https://doi.org/10.1109/TKDE.2004.33}
\showDOI{\tempurl}


\bibitem[\protect\citeauthoryear{Kuramochi and Karypis}{Kuramochi and Karypis}{2005}]%
        {kuramochi:2005:finding}
\bibfield{author}{\bibinfo{person}{Michihiro Kuramochi} {and} \bibinfo{person}{George Karypis}.} \bibinfo{year}{2005}\natexlab{}.
\newblock \showarticletitle{Finding frequent patterns in a large sparse graph}.
\newblock \bibinfo{journal}{\emph{Data mining and knowledge discovery}} \bibinfo{volume}{11}, \bibinfo{number}{3} (\bibinfo{year}{2005}), \bibinfo{pages}{243--271}.
\newblock
\urldef\tempurl%
\url{https://doi.org/10.1007/s10618-005-0003-9}
\showDOI{\tempurl}


\bibitem[\protect\citeauthoryear{Le, Vo, Nguyen, Fujita, and Le}{Le et~al\mbox{.}}{2020}]%
        {le:2020:mining}
\bibfield{author}{\bibinfo{person}{Ngoc-Thao Le}, \bibinfo{person}{Bay Vo}, \bibinfo{person}{Lam~BQ Nguyen}, \bibinfo{person}{Hamido Fujita}, {and} \bibinfo{person}{Bac Le}.} \bibinfo{year}{2020}\natexlab{}.
\newblock \showarticletitle{Mining weighted subgraphs in a single large graph}.
\newblock \bibinfo{journal}{\emph{Information Sciences}}  \bibinfo{volume}{514} (\bibinfo{year}{2020}), \bibinfo{pages}{149--165}.
\newblock
\urldef\tempurl%
\url{https://doi.org/10.1016/j.ins.2019.12.010}
\showDOI{\tempurl}


\bibitem[\protect\citeauthoryear{Leskovec, Huttenlocher, and Kleinberg}{Leskovec et~al\mbox{.}}{2010a}]%
        {leskovec:2010:predicting}
\bibfield{author}{\bibinfo{person}{Jure Leskovec}, \bibinfo{person}{Daniel Huttenlocher}, {and} \bibinfo{person}{Jon Kleinberg}.} \bibinfo{year}{2010}\natexlab{a}.
\newblock \showarticletitle{Predicting positive and negative links in online social networks}. In \bibinfo{booktitle}{\emph{Proceedings of the 19th international conference on World wide web}}. \bibinfo{publisher}{Association for Computing Machinery}, \bibinfo{address}{New York, NY, USA}, \bibinfo{pages}{641--650}.
\newblock


\bibitem[\protect\citeauthoryear{Leskovec, Huttenlocher, and Kleinberg}{Leskovec et~al\mbox{.}}{2010b}]%
        {leskovec:2010:signed}
\bibfield{author}{\bibinfo{person}{Jure Leskovec}, \bibinfo{person}{Daniel Huttenlocher}, {and} \bibinfo{person}{Jon Kleinberg}.} \bibinfo{year}{2010}\natexlab{b}.
\newblock \showarticletitle{Signed networks in social media}. In \bibinfo{booktitle}{\emph{Proceedings of the SIGCHI conference on human factors in computing systems}}. \bibinfo{publisher}{Association for Computing Machinery}, \bibinfo{address}{New York, NY, USA}, \bibinfo{pages}{1361--1370}.
\newblock


\bibitem[\protect\citeauthoryear{Leskovec, Kleinberg, and Faloutsos}{Leskovec et~al\mbox{.}}{2007}]%
        {leskovec:2007:graph}
\bibfield{author}{\bibinfo{person}{Jure Leskovec}, \bibinfo{person}{Jon Kleinberg}, {and} \bibinfo{person}{Christos Faloutsos}.} \bibinfo{year}{2007}\natexlab{}.
\newblock \showarticletitle{Graph evolution: Densification and shrinking diameters}.
\newblock \bibinfo{journal}{\emph{ACM transactions on Knowledge Discovery from Data (TKDD)}} \bibinfo{volume}{1}, \bibinfo{number}{1} (\bibinfo{year}{2007}), \bibinfo{pages}{2--es}.
\newblock


\bibitem[\protect\citeauthoryear{Leskovec, Lang, Dasgupta, and Mahoney}{Leskovec et~al\mbox{.}}{2009}]%
        {leskovec:2009:community}
\bibfield{author}{\bibinfo{person}{Jure Leskovec}, \bibinfo{person}{Kevin~J Lang}, \bibinfo{person}{Anirban Dasgupta}, {and} \bibinfo{person}{Michael~W Mahoney}.} \bibinfo{year}{2009}\natexlab{}.
\newblock \showarticletitle{Community structure in large networks: Natural cluster sizes and the absence of large well-defined clusters}.
\newblock \bibinfo{journal}{\emph{Internet Mathematics}} \bibinfo{volume}{6}, \bibinfo{number}{1} (\bibinfo{year}{2009}), \bibinfo{pages}{29--123}.
\newblock


\bibitem[\protect\citeauthoryear{Liu and Li}{Liu and Li}{2022}]%
        {liu:2022:satmargin}
\bibfield{author}{\bibinfo{person}{Muyi Liu} {and} \bibinfo{person}{Pan Li}.} \bibinfo{year}{2022}\natexlab{}.
\newblock \showarticletitle{SATMargin: Practical Maximal Frequent Subgraph Mining via Margin Space Sampling}. In \bibinfo{booktitle}{\emph{Proceedings of the ACM Web Conference 2022}}. ACM, \bibinfo{publisher}{Association for Computing Machinery}, \bibinfo{address}{New York, NY, USA}, \bibinfo{pages}{1495--1505}.
\newblock
\urldef\tempurl%
\url{https://doi.org/10.1145/3485447.3512196}
\showDOI{\tempurl}


\bibitem[\protect\citeauthoryear{Martinelli, Saracino, and Sgandurra}{Martinelli et~al\mbox{.}}{2013}]%
        {martinelli:2013:classifying}
\bibfield{author}{\bibinfo{person}{Fabio Martinelli}, \bibinfo{person}{Andrea Saracino}, {and} \bibinfo{person}{Daniele Sgandurra}.} \bibinfo{year}{2013}\natexlab{}.
\newblock \showarticletitle{Classifying android malware through subgraph mining}. In \bibinfo{booktitle}{\emph{International Workshop on Data Privacy Management}}. \bibinfo{publisher}{Springer}, \bibinfo{address}{Berlin, Heidelberg}, \bibinfo{pages}{268--283}.
\newblock


\bibitem[\protect\citeauthoryear{Mawhirter, Reinehr, Holmes, Liu, and Wu}{Mawhirter et~al\mbox{.}}{2021}]%
        {mawhirter:2021:graphzero}
\bibfield{author}{\bibinfo{person}{Daniel Mawhirter}, \bibinfo{person}{Sam Reinehr}, \bibinfo{person}{Connor Holmes}, \bibinfo{person}{Tongping Liu}, {and} \bibinfo{person}{Bo Wu}.} \bibinfo{year}{2021}\natexlab{}.
\newblock \showarticletitle{Graphzero: A high-performance subgraph matching system}.
\newblock \bibinfo{journal}{\emph{ACM SIGOPS Operating Systems Review}} \bibinfo{volume}{55}, \bibinfo{number}{1} (\bibinfo{year}{2021}), \bibinfo{pages}{21--37}.
\newblock


\bibitem[\protect\citeauthoryear{Mawhirter and Wu}{Mawhirter and Wu}{2019}]%
        {mawhirter:2019:automine}
\bibfield{author}{\bibinfo{person}{Daniel Mawhirter} {and} \bibinfo{person}{Bo Wu}.} \bibinfo{year}{2019}\natexlab{}.
\newblock \showarticletitle{Automine: harmonizing high-level abstraction and high performance for graph mining}. In \bibinfo{booktitle}{\emph{Proceedings of the 27th ACM Symposium on Operating Systems Principles}}. \bibinfo{publisher}{Association for Computing Machinery}, \bibinfo{address}{New York, NY, USA}, \bibinfo{pages}{509--523}.
\newblock


\bibitem[\protect\citeauthoryear{Mrzic, Meysman, Bittremieux, Moris, Cule, Goethals, and Laukens}{Mrzic et~al\mbox{.}}{2018}]%
        {mrzic:2018:grasping}
\bibfield{author}{\bibinfo{person}{Aida Mrzic}, \bibinfo{person}{Pieter Meysman}, \bibinfo{person}{Wout Bittremieux}, \bibinfo{person}{Pieter Moris}, \bibinfo{person}{Boris Cule}, \bibinfo{person}{Bart Goethals}, {and} \bibinfo{person}{Kris Laukens}.} \bibinfo{year}{2018}\natexlab{}.
\newblock \showarticletitle{Grasping frequent subgraph mining for bioinformatics applications}.
\newblock \bibinfo{journal}{\emph{BioData mining}} \bibinfo{volume}{11}, \bibinfo{number}{1} (\bibinfo{year}{2018}), \bibinfo{pages}{1--24}.
\newblock


\bibitem[\protect\citeauthoryear{Nguyen, Nguyen, Zelinka, Snasel, Nguyen, and Vo}{Nguyen et~al\mbox{.}}{2021}]%
        {nguyen:2021:method}
\bibfield{author}{\bibinfo{person}{Lam~BQ Nguyen}, \bibinfo{person}{Loan~TT Nguyen}, \bibinfo{person}{Ivan Zelinka}, \bibinfo{person}{Vaclav Snasel}, \bibinfo{person}{Hung~Son Nguyen}, {and} \bibinfo{person}{Bay Vo}.} \bibinfo{year}{2021}\natexlab{}.
\newblock \showarticletitle{A method for closed frequent subgraph mining in a single large graph}.
\newblock \bibinfo{journal}{\emph{IEEE Access}}  \bibinfo{volume}{9} (\bibinfo{year}{2021}), \bibinfo{pages}{165719--165733}.
\newblock
\urldef\tempurl%
\url{https://doi.org/10.1109/ACCESS.2021.3133666}
\showDOI{\tempurl}


\bibitem[\protect\citeauthoryear{Nguyen, Vo, Le, Snasel, and Zelinka}{Nguyen et~al\mbox{.}}{2020}]%
        {nguyen:2020:fast}
\bibfield{author}{\bibinfo{person}{Lam~BQ Nguyen}, \bibinfo{person}{Bay Vo}, \bibinfo{person}{Ngoc-Thao Le}, \bibinfo{person}{Vaclav Snasel}, {and} \bibinfo{person}{Ivan Zelinka}.} \bibinfo{year}{2020}\natexlab{}.
\newblock \showarticletitle{Fast and scalable algorithms for mining subgraphs in a single large graph}.
\newblock \bibinfo{journal}{\emph{Engineering Applications of Artificial Intelligence}}  \bibinfo{volume}{90} (\bibinfo{year}{2020}), \bibinfo{pages}{103539}.
\newblock
\urldef\tempurl%
\url{https://doi.org/10.1016/j.engappai.2020.103539}
\showDOI{\tempurl}


\bibitem[\protect\citeauthoryear{Nguyen, Zelinka, Snasel, Nguyen, and Vo}{Nguyen et~al\mbox{.}}{2022}]%
        {nguyen:2022:subgraph}
\bibfield{author}{\bibinfo{person}{Lam~BQ Nguyen}, \bibinfo{person}{Ivan Zelinka}, \bibinfo{person}{Vaclav Snasel}, \bibinfo{person}{Loan~TT Nguyen}, {and} \bibinfo{person}{Bay Vo}.} \bibinfo{year}{2022}\natexlab{}.
\newblock \showarticletitle{Subgraph mining in a large graph: A review}.
\newblock \bibinfo{journal}{\emph{Wiley Interdisciplinary Reviews: Data Mining and Knowledge Discovery}} \bibinfo{volume}{12}, \bibinfo{number}{4} (\bibinfo{year}{2022}), \bibinfo{pages}{e1454}.
\newblock


\bibitem[\protect\citeauthoryear{Richardson, Agrawal, and Domingos}{Richardson et~al\mbox{.}}{2003}]%
        {richardson:2003:trust}
\bibfield{author}{\bibinfo{person}{Matthew Richardson}, \bibinfo{person}{Rakesh Agrawal}, {and} \bibinfo{person}{Pedro Domingos}.} \bibinfo{year}{2003}\natexlab{}.
\newblock \showarticletitle{Trust management for the semantic web}. In \bibinfo{booktitle}{\emph{International semantic Web conference}}. Springer, \bibinfo{publisher}{Springer Berlin Heidelberg}, \bibinfo{address}{Berlin, Heidelberg}, \bibinfo{pages}{351--368}.
\newblock


\bibitem[\protect\citeauthoryear{RIPEANU, IAMNITCHI, and FOSTER}{RIPEANU et~al\mbox{.}}{2002}]%
        {ripeanu:2002:mapping}
\bibfield{author}{\bibinfo{person}{Matei RIPEANU}, \bibinfo{person}{Adriana IAMNITCHI}, {and} \bibinfo{person}{Ian FOSTER}.} \bibinfo{year}{2002}\natexlab{}.
\newblock \showarticletitle{Mapping the Gnutella network}.
\newblock \bibinfo{journal}{\emph{IEEE internet computing}} \bibinfo{volume}{6}, \bibinfo{number}{1} (\bibinfo{year}{2002}), \bibinfo{pages}{50--57}.
\newblock


\bibitem[\protect\citeauthoryear{Salem, Alokshiya, and Hasan}{Salem et~al\mbox{.}}{2021}]%
        {salem:2021:rasma}
\bibfield{author}{\bibinfo{person}{Saeed Salem}, \bibinfo{person}{Mohammed Alokshiya}, {and} \bibinfo{person}{Mohammad~Al Hasan}.} \bibinfo{year}{2021}\natexlab{}.
\newblock \showarticletitle{RASMA: a reverse search algorithm for mining maximal frequent subgraphs}.
\newblock \bibinfo{journal}{\emph{BioData Mining}}  \bibinfo{volume}{14} (\bibinfo{year}{2021}), \bibinfo{pages}{1--23}.
\newblock


\bibitem[\protect\citeauthoryear{Shi, Zhai, Xu, and Zhai}{Shi et~al\mbox{.}}{2020}]%
        {shi:2020:graphpi}
\bibfield{author}{\bibinfo{person}{Tianhui Shi}, \bibinfo{person}{Mingshu Zhai}, \bibinfo{person}{Yi Xu}, {and} \bibinfo{person}{Jidong Zhai}.} \bibinfo{year}{2020}\natexlab{}.
\newblock \showarticletitle{Graphpi: High performance graph pattern matching through effective redundancy elimination}. In \bibinfo{booktitle}{\emph{SC20: International Conference for High Performance Computing, Networking, Storage and Analysis}}. IEEE, \bibinfo{publisher}{Insitute of Electrical and Electronics Engineers}, \bibinfo{address}{Piscataway, NJ, USA}, \bibinfo{pages}{1--14}.
\newblock


\bibitem[\protect\citeauthoryear{Soos, Biere, Heule, Jarvisalo, and Suda}{Soos et~al\mbox{.}}{2019}]%
        {soos:2019:cryptominisat}
\bibfield{author}{\bibinfo{person}{Mate Soos}, \bibinfo{person}{Armin Biere}, \bibinfo{person}{M Heule}, \bibinfo{person}{M Jarvisalo}, {and} \bibinfo{person}{M Suda}.} \bibinfo{year}{2019}\natexlab{}.
\newblock \showarticletitle{CryptoMiniSat 5.6 with YalSAT at the SAT Race 2019}.
\newblock \bibinfo{journal}{\emph{Proc. of SAT Race}}  \bibinfo{volume}{B-2019-1} (\bibinfo{year}{2019}), \bibinfo{pages}{14--15}.
\newblock


\bibitem[\protect\citeauthoryear{Teixeira, Fonseca, Serafini, Siganos, Zaki, and Aboulnaga}{Teixeira et~al\mbox{.}}{2015}]%
        {teixeira:2015:arabesque}
\bibfield{author}{\bibinfo{person}{Carlos~HC Teixeira}, \bibinfo{person}{Alexandre~J Fonseca}, \bibinfo{person}{Marco Serafini}, \bibinfo{person}{Georgos Siganos}, \bibinfo{person}{Mohammed~J Zaki}, {and} \bibinfo{person}{Ashraf Aboulnaga}.} \bibinfo{year}{2015}\natexlab{}.
\newblock \showarticletitle{Arabesque: a system for distributed graph mining}. In \bibinfo{booktitle}{\emph{Proceedings of the 25th Symposium on Operating Systems Principles}}. \bibinfo{publisher}{Association for Computing Machinery}, \bibinfo{address}{New York, NY, USA}, \bibinfo{pages}{425--440}.
\newblock


\bibitem[\protect\citeauthoryear{Thomas, Valluri, and Karlapalem}{Thomas et~al\mbox{.}}{2010}]%
        {thomas:2010:margin}
\bibfield{author}{\bibinfo{person}{Lini~T Thomas}, \bibinfo{person}{Satyanarayana~R Valluri}, {and} \bibinfo{person}{Kamalakar Karlapalem}.} \bibinfo{year}{2010}\natexlab{}.
\newblock \showarticletitle{Margin: Maximal frequent subgraph mining}.
\newblock \bibinfo{journal}{\emph{ACM Transactions on Knowledge Discovery from Data (TKDD)}} \bibinfo{volume}{4}, \bibinfo{number}{3} (\bibinfo{year}{2010}), \bibinfo{pages}{1--42}.
\newblock
\urldef\tempurl%
\url{https://doi.org/10.1145/1839490.1839491}
\showDOI{\tempurl}


\bibitem[\protect\citeauthoryear{Yan and Han}{Yan and Han}{2002}]%
        {yan:2002:gspan}
\bibfield{author}{\bibinfo{person}{Xifeng Yan} {and} \bibinfo{person}{Jiawei Han}.} \bibinfo{year}{2002}\natexlab{}.
\newblock \showarticletitle{gspan: Graph-based substructure pattern mining}. In \bibinfo{booktitle}{\emph{2002 IEEE International Conference on Data Mining, 2002. Proceedings.}} IEEE, \bibinfo{publisher}{Institute of Electrical and Electronics Engineers}, \bibinfo{address}{Piscataway, NJ, USA}, \bibinfo{pages}{721--724}.
\newblock
\urldef\tempurl%
\url{https://doi.org/10.1109/ICDM.2002.1184038}
\showDOI{\tempurl}


\bibitem[\protect\citeauthoryear{Yao, Zheng, Zeng, Huang, Gui, Liao, Jin, and Xue}{Yao et~al\mbox{.}}{2020}]%
        {yao:2020:locality}
\bibfield{author}{\bibinfo{person}{Pengcheng Yao}, \bibinfo{person}{Long Zheng}, \bibinfo{person}{Zhen Zeng}, \bibinfo{person}{Yu Huang}, \bibinfo{person}{Chuangyi Gui}, \bibinfo{person}{Xiaofei Liao}, \bibinfo{person}{Hai Jin}, {and} \bibinfo{person}{Jingling Xue}.} \bibinfo{year}{2020}\natexlab{}.
\newblock \showarticletitle{A locality-aware energy-efficient accelerator for graph mining applications}. In \bibinfo{booktitle}{\emph{2020 53rd Annual IEEE/ACM International Symposium on Microarchitecture (MICRO)}}. IEEE, \bibinfo{publisher}{Institute of Electrical and Electronics Engineers}, \bibinfo{address}{Piscataway, NJ, USA}, \bibinfo{pages}{895--907}.
\newblock


\bibitem[\protect\citeauthoryear{Yuan, Yan, Qu, Adhikari, Khalil, Long, and Wang}{Yuan et~al\mbox{.}}{2023}]%
        {yuan:2023:tfsm}
\bibfield{author}{\bibinfo{person}{Lyuheng Yuan}, \bibinfo{person}{Da Yan}, \bibinfo{person}{Wenwen Qu}, \bibinfo{person}{Saugat Adhikari}, \bibinfo{person}{Jalal Khalil}, \bibinfo{person}{Cheng Long}, {and} \bibinfo{person}{Xiaoling Wang}.} \bibinfo{year}{2023}\natexlab{}.
\newblock \showarticletitle{T-FSM: A Task-Based System for Massively Parallel Frequent Subgraph Pattern Mining from a Big Graph}.
\newblock \bibinfo{journal}{\emph{Proc. ACM Manag. Data}} \bibinfo{volume}{1}, \bibinfo{number}{1}, Article \bibinfo{articleno}{74} (\bibinfo{date}{may} \bibinfo{year}{2023}), \bibinfo{numpages}{26}~pages.
\newblock
\urldef\tempurl%
\url{https://doi.org/10.1145/3588928}
\showDOI{\tempurl}


\bibitem[\protect\citeauthoryear{Zhao, Zhang, Xu, Zheng, and Guo}{Zhao et~al\mbox{.}}{2020}]%
        {zhao:2020:kaleido}
\bibfield{author}{\bibinfo{person}{Cheng Zhao}, \bibinfo{person}{Zhibin Zhang}, \bibinfo{person}{Peng Xu}, \bibinfo{person}{Tianqi Zheng}, {and} \bibinfo{person}{Jiafeng Guo}.} \bibinfo{year}{2020}\natexlab{}.
\newblock \showarticletitle{Kaleido: An efficient out-of-core graph mining system on A single machine}. In \bibinfo{booktitle}{\emph{2020 IEEE 36th International Conference on Data Engineering (ICDE)}}. IEEE, \bibinfo{publisher}{Insitute of Electrical and Electronics Engineers}, \bibinfo{address}{Piscataway, NJ, USA}, \bibinfo{pages}{673--684}.
\newblock


\end{thebibliography}

\end{document}